\documentclass{article}

\usepackage{arxiv}

\usepackage[svgnames]{xcolor} %
\usepackage[T1]{fontenc}    %

\usepackage{graphicx}

\usepackage{xspace}

\usepackage[english]{babel} %

\usepackage[style=english]{csquotes} %

\usepackage{mdframed} %

\usepackage{amsmath}
\usepackage{amssymb}
\usepackage{stmaryrd}

\usepackage{ebproof} %
\newcommand\hypo{\Hypo} %
\newcommand\infer{\Infer}

\usepackage[cmintegrals,cmbraces]{newtxmath} %

\usepackage[colorlinks,allcolors=Blue]{hyperref}

\usepackage{subfig}

\usepackage{booktabs}
\usepackage{multirow}
\usepackage{paralist}

\usepackage{tikz}
\usetikzlibrary{cd,calc,automata,shapes}

\usepackage[authoryear,sort&compress]{natbib}

\newcommand{\tran}[4][]{#2 \xrightarrow{#3}_{#1} #4}
\newcommand{\ntran}[4][]{#2 \stackrel{#3}{\nrightarrow}_{#1} #4}

\newcommand{\refinedBy}[1][T]{\sqsubseteq_{#1}}

\newcommand{\traces}[2][]{\mathit{traces}_{#1}(#2)}
\newcommand{\trace}[1][]{\langle #1 \rangle}
\newcommand{\event}[1][]{\mathit{#1}}
\newcommand{\action}[2][]{\event[#2]_{#1}}
\newcommand{\acend}[2][]{\action[#1]{e}^{\rfct[#2]}}
\newcommand{\acendr}[2][]{\action[#1]{\overline{e}}^{\rfct[#2]}}
\newcommand{\acmit}[2][]{\action[#1]{m}^{\rfct[#2]}}
\newcommand{\acopr}[2][]{\action[#1]{o}^{\rfct[#2]}}
\newcommand{\prCh}[2][x]{\,?{#1}:{#2} \operatorname{\rightarrow}}
\newcommand{\genpar}[1][]{\mathop{\parallel}\limits_{#1}}

\newcommand{\rfct}[1][f]{\mathsf{#1}}
\newcommand{\phase}[2][\sigma]{#1^{(#2)}}
\newcommand{\inact}[1][\rfct]{0^{#1}}
\newcommand{\activ}[1][f]{\mathit{#1}}
\newcommand{\mitig}[1][\rfct]{\overline{#1}}

\tikzset{
  blockdiag/.style={
    dim/.style={align=flush left},
    tit/.style={align=left,anchor=south west},
    itm/.style={align=center,minimum height=3em,minimum width=3cm,fill=gray!20},
    lab/.style={align=center,midway}
  },
  machine/.style={
    inner sep=2pt,
    outer sep=0pt,
    node distance=2cm,
    every edge/.style={draw,black,-{Stealth}},
    every loop/.style={-{Stealth}}
  },
  riskstr/.style={
    inner sep=2pt,
    outer sep=0pt,
    node distance=2cm,
    every edge/.style={very thick,draw,black,-{Stealth}},
    every loop/.style={-{Stealth}}
  },
  prjnc/.style={circle,draw},
  st/.style={circle,fill=black!20,align=center},
  gst/.style={circle,fill=ForestGreen!50,align=center},
  bst/.style={circle,fill=red!50,align=center},
  dynev/.style={sloped,above,midway},
  dynevin/.style={sloped,above,midway},
  dynevout/.style={sloped,above,midway},
  logev/.style={below,dashed,midway},
  hybev/.style={below,very thick,gray,align=center,midway},
  label/.style={align=center,midway},
  endg/.style={red},
  mitg/.style={ForestGreen}
}

\newmdtheoremenv[%
hidealllines=true,
leftmargin=20,%
rightmargin=20,
backgroundcolor=gray!30,%
settings={\small},
ntheorem]{example}{Example}
\newtheorem{definition}{Definition}%
\newtheorem{remark}{Remark}%
\newtheorem{theorem}{Theorem}%
\newtheorem{proof}{Proof}%
\newtheorem{corollary}{Corollary}%
\newtheorem{lemma}{Lemma}%

\usepackage{cleveref}
  \Crefname{theorem}{Theorem}{Theorems}
  \Crefname{lemma}{Lemma}{Lemmas}
  \Crefname{property}{Property}{Properties}
  \Crefname{equation}{Formula}{Formulas}
  \Crefname{proof}{Proof}{Proofs}

\tikzset{
  riskm/.pic={
    \begin{scope}[riskstr,node distance=3cm]
      \footnotesize
      \node[initial,gst] (s1) {\checkmark\\$\{s_1,s_3$,\\$s_5\}$};
      \node[bst] (s3) [right of=s1] {$\lightning$\\$\{s_2,s_4\}$};
      \node[st] (s6) [right of=s3] {?\\$\{s_6\}$};
      \draw
      (s1) edge[dynev,endg,bend left] node {$\{in_1,out_1\}$} (s3)
      (s3) edge[dynev,mitg,bend left] node[below] {$\{in_2,out_3\}$} (s1)
      (s1) edge[hybev,loop above] node {$\{in_1,out_2$,\\$act_1,act_3\}$} (s1)
      (s3) edge[hybev,dotted] node {$\{act_5\}$} (s6)
      (s3) edge[hybev,loop above] node {$\{in_2,out_4$,\\$act_4\}$} (s3);
    \end{scope}
  },
  procm/.pic={
    \begin{scope}[machine]
      \footnotesize
      \node[initial,st] (s1) {$s_1$};
      \node[prjnc] (j1) [right of=s1] {};
      \node[st] (s2) at ($(j1)+(2,.75)$) {$s_2$};
      \node[st] (s3) at ($(j1)+(2,-.75)$) {$s_3$};
      \node[prjnc] (j2) [right of=s2] {};
      \node[st] (s4) at ($(j2)+(2,-.75)$) {$s_4$};
      \node[st] (s5) at ($(j2)+(2,.75)$) {$s_5$};
      \node[st] (s6) at ($(s3)+(2,0)$) {$s_6$};
      \draw
      (s1) edge[dynev] node {$in_1$} (j1)
      (j1) edge[dynev] node {$\Lambda_1\colon out_1$} (s2)
      (j1) edge[dynev,below] node {$1-\Lambda_1\colon out_2$} (s3)
      (s3) edge[logev,sloped,above,bend left] node {$act_1$} (s1)
      (s1) edge[logev,loop below] node {$act_2$} (s1)
      (s2) edge[dynev] node {$in_2$} (j2)
      (j2) edge[dynev] node {$\Lambda_2\colon out_3$} (s5)
      (j2) edge[dynev,below] node {$1-\Lambda_2\colon out_4$} (s4)
      (s5) edge[logev,bend right,sloped] node {$act_3$} (s1)
      (s2) edge[logev,sloped] node {$act_5$} (s6)
      (s4) edge[logev,loop below] node {$act_4$} (s4);
    \end{scope}
  },
  abstlev/.pic={
    \begin{scope}[blockdiag]
      \node[itm] (riskm) at (0,0) {\textbf{Risk Model} $\mathfrak{R}$};
      \node[itm,draw,dotted,fill=white] (procm) at
      ($(riskm)+(0,-8em)$) {Process Model $P$};
      \draw[arrows={-latex},thick,tips=proper]
      (riskm) edge[dotted] node[lab,fill=white] {abstracts\\from} (procm)
      ;
    \end{scope}
  }
}

\begin{document}
\label{firstpage}

\title%
{Risk Structures: Towards Engineering Risk-aware
  Autonomous Systems}
\author%
{Mario Gleirscher\\
  Computer Science Department, University of York, York, UK%
\thanks{\emph{Correspondence and offprint requests to:}
  Mario Gleirscher, Univerity of York, Deramore Lane,
  Heslington, York YO10 5GH, UK.  e-mail:
  mario.gleirscher@york.ac.uk
  \newline
  This work is supported by the Deutsche Forschungsgemeinschaft~(DFG) under
  Grants~no.~GL~915/1-1 and~GL~915/1-2. 
  \copyright\ 2019. This manuscript is
  made available under the CC-BY-NC-ND 4.0 license
  \url{http://creativecommons.org/licenses/by-nc-nd/4.0/}.
  \newline
  \textbf{Reference Format:}
  Gleirscher, M.. \emph{Risk Structures: Towards Engineering Risk-aware
  Autonomous Systems} (\today). Unpublished working paper. 
  Department of Computer Science, University of York, United
  Kingdom.
}}
\date{\today}

\maketitle

\begin{abstract}
  Inspired by widely-used techniques of causal modelling in risk,
  failure, and accident analysis, this work discusses a com positional
  framework for risk modelling.  Risk models capture fragments of the
  space of risky events likely to occur when operating a machine in a
  given environment.  Moreover, one can build such models into
  machines such as autonomous robots, to equip them with the ability
  of risk-aware perception, monitoring, decision making, and control.
  With the notion of a risk factor as the modelling primitive, the
  framework provides several means to construct and shape risk models.
  Relational and algebraic properties are investigated and proofs
  support the validity and consistency of these properties over the
  corresponding models.  Several examples throughout the discussion
  illustrate the applicability of the concepts.  Overall, this work
  focuses on the qualitative treatment of risk with the outlook of
  transferring these results to probabilistic refinements of the
  discussed framework.
\end{abstract}

\tableofcontents

\section{Introduction}
\label{sec:introduction}

For surgical robot assistance, the summary of safety mechanisms by
\citet{Howe1999-RoboticsSurgery} ranges from active and passive
mechanisms over independent safety monitors to supervisory control.
This summary gives an impression of the number of risk factors to be
handled by dedicated mechanisms in safety-critical robots.  Moreover,
\citeauthor{Howe1999-RoboticsSurgery} already suggested that full
robot autonomy requires improved sensor fusion and qualitative
reasoning.
Pushing this idea much further,
\citet{Holland2003-RobotsInternalModels}
discuss how autonomous robots might be able to internalise models of
consciousness.  Inspired by this direction, this work aims at the use
of finite symbolic %
models to integrate a form of consciousness of risk in such systems,
hereafter called \emph{risk awareness}.  To achieve this, let us first
revisit some basics of the construction of failure and risk models.

Highly automated machines can be both faulty and engage in dangerous
events and are expected to handle faults and dangerous events
automatically.
\emph{Hazards} generalise the concept of faults, errors, and failures
to any kind of dangerous events.  The notion of \emph{risk} then
allows to discuss temporal and causal relationships between such
events.
Some hazards can be avoided by design.  However, for many
hazards, only the \emph{likelihood of their occurrence} or the
\emph{severity of their consequences} can be reduced.
Engineers use specific models to analyse these hazards and to achieve
their reduction.  In formal verification, we identify the models we
accept, the ones we do not accept, and compare models to decide which
ones we accept more than others.
In this work, we consider systems that can fail or engage in dangerous
events.  With the help of models, we study how we can handle such
events.  Let us further motivate this with examples.

\begin{example}[Qualitative Evaluation of the Risk of Functional Failures]
  \label{exa:evalfuncfault}
  In a car airbag, the event ``airbag release'' is associated with two
  general functional hazards: \emph{failure on demand}~(i.e.,\xspace action
  not performed when requested or when its guard is enabled) and
  \emph{spurious trip}~(i.e.,\xspace action performed when not requested or
  when its guard is not enabled).  The risk incorporated by these
  hazards depends on the current state of the driving process.  For
  each of these hazards, we can separate the driving process into
  \emph{situations}.  This step yields a table whose cells can be
  filled with \emph{risk information}, particularly, an analysis of
  the probability
  and/or the consequences of an event occurrence:\\
  \begin{tabular}{lp{5cm}p{5cm}}
    \toprule
    Context: Vehicle in \dots
    & \textbf{Spurious trip}
    & \textbf{Failure on demand}
    \\
    \multicolumn{3}{c}{\dots of \textbf{airbag}:}
    \\\midrule
    \dots \textbf{manual mode}
    & 1. \emph{consequences from distraction \& bag shock}
    & (irrelevant) \\
    \dots \textbf{autonomous mode}
    & 2. \emph{consequences from bag shock}
    & (irrelevant) \\
    \dots \textbf{collision}
    & (irrelevant)
    & 3. \emph{consequences from crash without airbag}
    \\\bottomrule
  \end{tabular}\\
  Whereas case 1 can lead to a fatal car crash due to \emph{loss of
    control}~(i.e.,\xspace by distraction) %
  of the driving process, case 2 as a generalisation of case
  1 may cause serious injuries but unlikely a fatal crash.  Case 3
  is different, though the loss of control %
  is irrelevant and risk %
  is inherited from situations without airbag.
\end{example}

The number of \emph{contexts} or \emph{situations}, \emph{functions},
and \emph{hazards} requires a large number of these tables to be
evaluated at design time or during operation.  Additionally, these
tables are related according to dependencies between situations,
functions, and hazards~(e.g.\xspace loss of human control \emph{requires}
spurious trip in manual mode, risk analysis of the airbag \emph{refers
  to} crash risk analysis without an airbag, risk analysis in manual
mode is an \emph{extension of} risk analysis in autonomous mode).

\begin{example}[Qualitative Evaluation of the Risk of Operational Incidents]
  \label{exa:evaloperincid}
  Consider a collision of a vehicle with another object.  The overall
  risk depends on the likelihood of a collision from the
  current state and the possible consequences of the collision.\\
  \begin{tabular}{lp{4.5cm}p{5cm}}
    \toprule
    Context: Process with \dots 
    & \textbf{Near-collision}
    & \textbf{Collision} \\
    \multicolumn{3}{c}{\dots of \textbf{vehicle}:}
    \\\midrule
    \dots \textbf{following vehicle} %
    & probability of collision
    & consequences of passive collision %
    \\ 
    \dots \textbf{leading vehicle} %
    & probability of collision
    & consequences of rear-end collision
    \\
    \dots \textbf{oncoming vehicle}
    & probability of collision
    & consequences of head-on collision
    \\\bottomrule
  \end{tabular}
\end{example}

\Cref{exa:evaloperincid} illustrates the probabilistic relationship
between the two events \emph{near-collision} and \emph{collision}, and
the relationship between operational incidents and functions, that is,\xspace the
airbag as the subject of risk assessment in \Cref{exa:evalfuncfault}
is the safety function mitigating risk after collision events.
Note the difference between near-collision and the three other
hazards, collision, and spurious trip and failure on demand of the
airbag.  The three latter form \emph{hard}
events inasmuch as the focus lies on \emph{consequence estimation}
whereas near-collision can be understood as a \emph{soft}
event where the focus of risk assessment lies on \emph{probability
  estimation}.

\subsection{Abstractions for Machine Safety %
  and Risk Awareness} %
\label{sec:abstraction}

Risk assessment of an autonomous robot encompasses
\begin{itemize}
\item the analysis of chains of undesired events the machine can
  engage in, qualitatively~(i.e.,\xspace from a causal viewpoint) and
  quantitatively~(e.g.\xspace from a probabilistic viewpoint), and
\item the analysis of what the machine is capable of~(i.e.,\xspace from a
  functional, situational, and performance viewpoint).
\end{itemize}
Autonomous robots have to handle many such event chains during
operation.  Hence, the tables mentioned above have to be identified
and pre-filled at design time~(e.g.\xspace by using qualitative risk
matrices).  Some of these tables will have to be continuously
re-assessed at run-time~(e.g.\xspace by prediction and quantitative
risk assessment).
Such robots need to continuously judge risk stemming from their past
and planned behaviours, using introspection, estimating the current
state, and predicting future states of the whole process.  In the
following, we will use the term \emph{process} to refer to the
operation of a robot in its physical environment.

The \Cref{exa:evalfuncfault,exa:evaloperincid} stimulate two questions
when designing run-time mitigation measures: \emph{Which situations,
  functions, and hazards are there?  How do we consider all these in a
  manageable run-time model?}
A run-time risk model identifies the undesirable or \emph{dangerous}
subset and the desirable or \emph{safe} subset of the state space of
the process.  Knowing the dangerous subset helps assessing the
measures for avoiding to reach this subset or for leaving this subset.
Knowing the safe subset helps assessing the measures for not leaving
this subset or for reentering this subset.
If one of these sets is completely identified, we can derive the other
one by set complement.  However, often a risk model helps labelling
only some fragments of safe and dangerous states such that a set of
\emph{unlabelled states} is left.  Moreover, instead of a dichotomous
scale~(i.e.,\xspace safe or dangerous), the risk model can help evaluating
risk per state on a cardinal scale~(e.g.\xspace \emph{risk level} of a state
as a continuous measure~\citep{Sanger2014-RiskAwareControl}), or
using fuzzy sets~(i.e.,\xspace degrees of membership of a state in both the
safe and dangerous sets).

Starting form unlabelled states, we can first be permissive, that is,\xspace
successively \emph{identify the dangerous subset}~(e.g.\xspace by estimating
risk levels) and reduce the safe subset accordingly.  Alternatively,
we can successively \emph{expand the safe subset}~(e.g.\xspace by estimating
risk levels) until we conclude unacceptable risk from our model.
While both approaches occur in safety cultures, the latter can
sometimes be too restrictive~(i.e.,\xspace the whole state space is a-priori
unsafe).
The following discussion will focus on the permissive approach.

Risk awareness results from the fact that an autonomous robot complies
with or refines the \emph{risk model} at run-time.
Consequently, our main questions are: \emph{What constitutes a
  powerful risk model to make autonomous robots risk-aware? Moreover,
  how can we systematically engineer a consistent and valid risk model
  for an autonomous robot?}
To separate concerns in the modelling of autonomous robots, we
consider two abstractions, the \emph{process model} $P$ and the
\emph{risk model} $\mathfrak{R}$ as shown in \Cref{fig:abslevels}.

$P$ captures the behaviours we might observe in the actual process,
for example,\xspace the behaviours that are generated from a robot in its environment
continuously making decisions, performing logical actions~(i.e.,\xspace
changing the data state, $act_i$, dotted arcs), stimulating physical
actions~(i.e.,\xspace generating control inputs, $in_i$, solid arcs), and
producing and observing outcomes~(i.e.,\xspace sensing process outputs,
$out_i$, solid arcs).  Actions and outcomes are the events of interest
in $P$.  Events and states~($s_i$) represent the \emph{observables}
to reason about behaviour.  Uncertainty in $P$ allows several
outcomes or successor states from one action~(e.g.\xspace $in_1,in_2$)
associated with parameters~($\Lambda_i$) forming probability
distributions on the outcomes.  A dynamical model of $P$~(e.g.\xspace
a hybrid automaton) can be used instead of an uncertainty model~(e.g.\xspace
a \textsc{Markov} decision process).

$\mathfrak{R}$ abstracts from $P$ and is comprised of a set of
\emph{risk factors}, each classifying $P$'s state space into a safe
region~(i.e.,\xspace green node signifying the desirable subset), a risky
region~(i.e.,\xspace red node signifying the undesirable subset), and an
unlabelled region~(grey node).  $\mathfrak{R}$ classifies and reduces the
event space~(i.e.,\xspace the \emph{alphabet}) of $P$ to events relevant for
risk assessment, that is,\xspace \emph{endangerments}~(red arcs) and
\emph{mitigations}~(green arcs).
We assume to establish observational \emph{refinement} or some
bisimilarity between $P$ and $\mathfrak{R}$.  In $\mathfrak{R}$, we have many choices
for abstraction, for example,\xspace we can focus on logical actions~($act_i$),
process responses~($out_i$), control inputs~($in_i$), or any
combination thereof.  We can also craft and compare several risk
models of the same process, each representing the view of a specific
risk analysis.
The remainder of this work will deal with a formal framework for the
systematic construction of consistent and valid risk models.

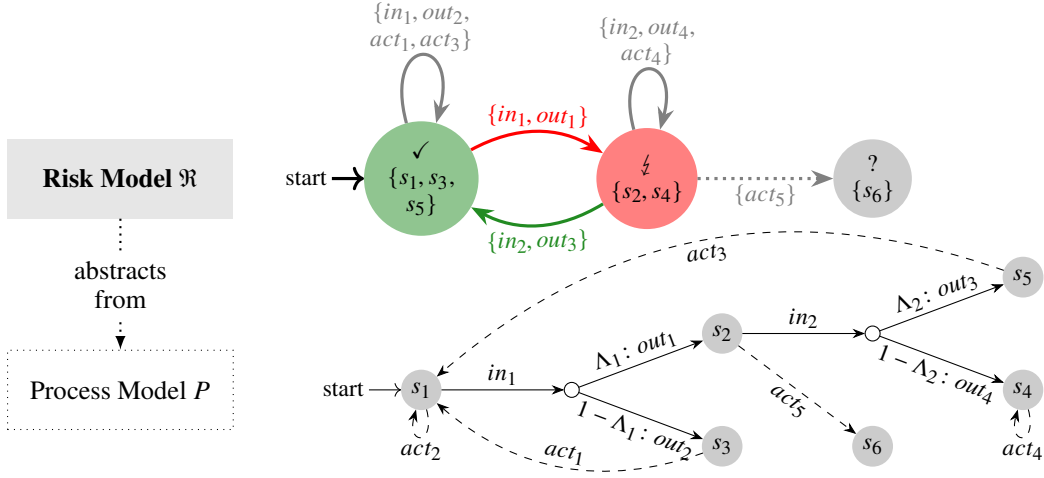
\begin{figure}
  \centering
  \begin{tikzpicture}
    \pic[at={(0,5)}]{abstlev};
    \pic[at={($(riskm)+(4,0)$)}]{riskm};
    \pic[at={($(procm)+(4,0)$)}]{procm};
  \end{tikzpicture}
  \caption{Two abstraction levels: Simple risk model $\mathfrak{R}$ with a
    single risk factor partitioning the state space of process model
    $P$ and, this way, forming a view of $P$ with respect to this
    specific risk factor
    \label{fig:abslevels}}
\end{figure}

\subsection{Related Work}
\label{sec:risks:relwork}

Here, we will discuss related work 
\begin{inparaitem}
\item in the area of algebraic methods for risk assessment,
\item in dependability of repairable systems, 
\item in safety monitoring, and
\item in risk-sensitive control and risk-aware planning.
\end{inparaitem}

\paragraph{Algebraic Methods in System Risk Assessment.}

To the best knowledge of the author, this work is the first to provide
an algebraic account of safety risk modelling for risk-aware
autonomous systems.
However, algebraic methods have been proposed for IT network security
risk assessment.
For example,\xspace\citet{Hamdi2003-Algebraicspecificationnetwork} formalise the
security risk management life cycle as an algebraic specification with
the aim of consistency checking of particular risk analyses viewed as
individual algebras.  Probability of occurrence and severity of
consequence are modelled as metrics over attacks~(i.e.,\xspace action sets) to
select optimal countermeasures for attacks using multi-objective
optimisation.  Although their approach focuses security risk
management at design-time, their framework is inspiring for the
further development of the work at hand.

\paragraph{Dependability Methods for Repairable Systems.}
\label{sec:depend-meth-repa}

Risk structures overlap in their potential use with works in the field
of dependability assessment.  This relationship gets visible, for example,\xspace in
the approach of \citet{Unanue2018-ExplicitModellingTreatment}.
The authors annotate a component architecture with a fault model,
synthesise a failure-repair
automaton, %
generate a temporal fault/repair tree, %
synthesise minimum cut sequences from this tree, and
construct an extended form of \textsc{Petri} net %
to calculate failure probabilities
from these sequences.
While they focus on pure failure assessment of generic systems,
risk structures generalise their approach by including severity of
consequences to enable qualitative reasoning about risk rather than
only about failure.

\paragraph{Safety Monitoring in Autonomous Systems.}
\label{sec:relwork:monitors}

One of the main intentions behind risk structures is their use as
active monitors and mechanisms for handling undesired events.
This has been an active areas in robotics research for many years.

\citet{Sobek1988-Integratedplanningexecution} propose a robot
architecture
where a safety monitor identifies obstacles not recognised in a
previous planning step.
Such events trigger local corrective actions, for example,\xspace obstacle avoidance
while a mitigation monitor checks for the success of a corrective
action to hand over to the main planner again.
Similarly, \citet{Simmons1994-Structuredcontrolautonomous}
speak of deliberative components to handle normal situations whereas
reactive
behaviours %
are activated to handle exceptional situations.

\citet{Guiochet2008-OnlineSafetyMonitoring} propose a risk model
using mode transition systems
where each mode represents a safety constraint
and modes are partially ordered.
Refining this approach,
\citet{Mekki-Mokhtar2012-SafetyTriggerConditions} distinguish between
safe, warning, and catastrophic states with safety trigger conditions.
Based on this framework, \citet{Machin2018-SMOFSafetyMOnitoring}
apply model checking to determine whether catastrophic states can be
reached.
The authors also present a tree-based algorithm for the synthesis of
mitigation strategies, that is,\xspace minimal sequences of control interventions
to reach the nearest safe state from any warning state and fulfilling
validity, permissiveness, and safety properties.
While their approach is useful for the refined modelling of hazards,
the presented work could enhance their framework with an algebraic
method.
Particularly, the work at hand allows to specify dependencies between
hazard-related variables. %

Based on the ``monitoring oriented programming''
approach~\citep{Meredith2011-overviewMOPruntime},
\citet{Huang2014-ROSRVRuntimeVerification} present a monitoring
infrastructure for checking trace properties~(specified in different
formalisms)
observable from communications between modules running on the robot
platform ROS.\footnote{Robot operating system, see
  \url{https://www.ros.org}.}  Their approach indicates how the risk
model presented here can be implemented as a monitor using their
framework.

\citet{Sorin2016-RulebasedDynamic} propose a framework for the
modelling and generation of safety monitors for ROS-based autonomous
robots using if-then rules and safety action specifications.  The
approach is tested with an automatic vehicle in a farm
environment. %
While their framework considers many important implementation aspects,
the presented risk model could add a methodological layer to their
approach, improving scalability and the formal treatment of
conflicting or interfering mitigation actions.

\paragraph{Risk-sensitive Control and Risk-aware Planning in
  Autonomous Systems.}
\label{sec:relwork:risksenscontr}

\citet{Althoff2007} discusses an approach to metric reachability
for linearisable dynamics.  Efficiency of the reachability analyser is
achieved by discretising the linearised model.  The authors apply this
approach to collision avoidance control in autonomous vehicles.  For
this, they discuss an abstraction of the discretised model into a
Markov chain where the pair-wise convolution of reachable sets serves
the calculation of collision risk of an ego vehicle with an
oncoming %
vehicle.

Several applications of stochastic optimal control %
aim at minimising collision risk of autonomous robots and vehicles.
\citet{Althoff2011-Safetyassessmentrobot} work with a probabilistic
version of inevitable collision
states~\citep{Fraichard2004-Inevitablecollisionstates} to
approximate collision probability and cost beyond the planning horizon
from Monte Carlo sampling of trajectories.  These
metrics allow the ranking of the simulated trajectories
in navigation decisions.
\citet{Pereira2013-RiskawarePath} experiment with a minimum
expected risk planner and a risk-aware Markov decision process in
autonomous underwater vehicle navigation exposed to perturbations by
ocean currents.
\citet{Sanger2014-RiskAwareControl} discusses a framework for
risk-aware control where ``movement is governed by estimates of risk
based on uncertainty about the current state and knowledge of the cost
of errors.''  The author illustrates an implementation of this
framework for autonomous driving by a neuronal network.
\citet{Feyzabadi2014-Riskawarepath} propose an efficient
risk-aware path planning algorithm using constrained Markov decision
processes, illustrating their approach by an autonomous indoor
navigation problem.
\citet{Muller2014-Riskawaretrajectory}
formalise collision risk by a Gamma distribution of the state and
uncertain %
distance to the nearest obstacle.

For discrete planning and navigation,
\citet{Shalev-Shwartz2018-FormalModelSafe} define a parametric
model
for the investigation of collision-free autonomous driving.  Their
model implements driving rules according to the Duty of Care approach
from Tort law,
assuming \emph{proper response} of all relevant traffic
participants %
in typical driving scenarios.
For efficient planning, the action space of the vehicle is discretised
to solve the corresponding %
optimisation problem.
The authors also determine the probability of sensing mistakes
after applying a triple redundancy pattern to the sensor sub-system.
\citet{Chen2018-RearEndCollision} discuss the prediction of
rear-end collision %
risk.  Based on the current vehicle state and the sensed
environment, a \textsc{Kalman} filter predicts the next state of the vehicle
and its environment at the end of a monitoring interval. %
The predicted state represents evidence in the %
\textsc{Bayes}ian network for the estimation of the collision
probability.  Low risk is translated into a warning, high risk into a
control intervention which, however, is not focused in their work.

\citet{McCausland2013-proactiveriskaware} investigate risk-driven
self-organisation of a communicating collective of mobile robots
acting as environment sensor nodes.  Their risk model is comprised of
three factors and represents a fixed node-local risk metric which is
continuously evaluated.

In all these works, risk represents either a~(chance) constraint not
to be violated or a minimisation criterion for determining an optimal
plan.  Stochastic models allow the per-state estimation of expected
risk.  The notion of severity interval in the work at hand would add
another possibly helpful uncertainty factor to such approaches and
allow the per-state estimation of an \emph{expected risk range}.
Complementary to these control-theoretic approaches, the work at hand
provides an algebra for constructing risk models~(i.e.,\xspace state space,
value function, action space), hence, allowing reasoning about the
composition of models of a variety of risks beyond collision and
relating these models via refinement.  This account addresses the
systematic construction and verification of risk models.  The
mentioned control approaches suggest implementations of controllers
from validated risk models.

\subsection{Contributions and Overview}
\label{sec:contribution}

This work contributes to the state of the art of formal engineering
methodology for highly automated safety-critical systems,
particularly, for the engineering of risk-aware robots and autonomous
systems, in several regards:
\begin{itemize}
\item Humans can qualitatively identify or predict dangerous
  situations and take actions to avoid or recover from such
  situations.  Based on that idea, this work approaches
  \emph{qualitative risk modelling}~(\Cref{sec:risk:model}),
  yet %
  making provisions to refine the presented model into quantitative
  risk models~(\Cref{sec:relwork:risksenscontr}).

\item It provides a \emph{conceptual framework} for risk modelling,
  formalises this framework, and proves \emph{algebraic} properties of
  the proposed
  concepts~(\Cref{sec:risk:model,sec:risk:mitorders,sec:risk:constraints,sec:risk-structures}),
  hence, going beyond the works summarised in
  \Cref{sec:risks:relwork}.  The framework %
  further develops the ideas in
  \citep{Gleirscher2018-Strukturenfurdie,Gleirscher2017-FVAV}.
\item Bridging the gap to works in \Cref{sec:relwork:monitors}, it
  investigates how \emph{safety} of machines such as autonomous robots
  can be evaluated at design time and at run-time based on the
  proposed framework~(\Cref{sec:risk:safety}).
\item It establishes a relationship to \emph{other techniques}
  conventionally and widely used in risk, failure, and accident
  analysis~(\Cref{sec:depend-meth-repa,sec:risk:constraints}).
\item It discusses several \emph{examples} to motivate the choice of
  these concepts~(\Crefrange{exa:evalfuncfault}{exa:depend1}) and
  explains the specifics and pitfalls of the abstraction to be made by
  users of the framework~(\Cref{rem:abstroforder1,rem:abstroforder2}).
\item Complementary to the works in \Cref{sec:relwork:risksenscontr},
  it characterises \emph{risk awareness} as the internalisation of
  expressive and \emph{structured} risk models by robot controllers to
  evaluate, predict, and memorise risk by both introspection and
  environment sensing~(\Cref{sec:risk-structures}).
\item It describes how risk models represent \emph{acceptance
    specifications}~(i.e.,\xspace the refinement of the risk model by the
  process, after hiding irrelevant events, expresses risk acceptance)
  and \emph{run-time monitors} that continuously decide whether and
  how the process violates safety %
  or achieves co-safety %
  (\Cref{sec:notes-monitoring}). %
\end{itemize}
The remainder of this article is structured as follows: After an
introduction to the field~(\Cref{sec:risks:background}) including the
formal preliminaries~(\Cref{sec:formal-preliminaries}), the
\Cref{sec:risk:model,sec:risk:mitorders,sec:risk:constraints,sec:risk-structures}
present the contribution followed by a brief discussion in
\Cref{sec:rap:discussion}.  The article concludes with
\Cref{sec:concl-future-work}.  Proof details are listed in
\Cref{sec:proofs}.

\section{Background}
\label{sec:risks:background}

This section provides some background on risk modelling and
assessment, failure analysis, and run-time monitoring.  It also draws
a relationship to the application domain of robots and autonomous
systems.

\subsection{Notions of Risk}
\label{sec:risk}

Risk can be characterised as the possibility of \emph{undesired}
outcomes~(e.g.\xspace hazardous events and states) of an action with several
alternative but \emph{uncertain
outcomes}~\citep{%
  Kaplan1981-QuantitativeDefinitionRisk}.
In systems engineering, risk is usually assessed by 
measuring
\begin{inparaitem}
\item the probability of hazards,
\item the severity of their consequences in case of their occurrence,
\item the probability of occurrence of these consequences~(e.g.\xspace an
  accident) after hazards have occurred, and
\item the exposure of the system to these
  hazards~\citep{Leveson1995-SafewareSystemSafety}.
\end{inparaitem}
Variations and simplifications of these measures are in use across
different disciplines, application domains, and corresponding
standards.  
For the estimation of these measures, it is necessary to understand
causal relationships between events.  Reasoning about \emph{causal
propositions} has been formalised in philosophy, mathematics, and
computer science~%
\citep{%
  Lewis1973-Causation%
}.
\emph{Probabilistic risk assessment}~(PRA) focuses on the use of
stochastic models to quantify such
propositions~\citep{Kumamoto2007-Satisfyingsafetygoals}.  Apart from
probabilistic models, a variety of \emph{uncertainty models} and
analysis techniques are used in engineering, for example,\xspace as summarised by
\citet{Oberguggenberger2015-Analysiscomputationhybrid} for structural
safety, the reduction of the risk of dangerous incidents in civil
engineering projects.

\subsection{The Risk of Undesired Events} 
\label{sec:relwork:undesiredevents}

\paragraph{System Accidents, Dangerous Incidents and Failures.}
\label{sec:relwork:accincfail}

Unfortunately, knowledge about risk mostly results from system
accidents.  To learn from such events, accident researchers use
methods such as AcciMaps~\citep{%
  Svedung2002-Graphicrepresentationaccident}.
For prevention and in response to accidents, methods such as
hazard operability studies~(HazOp), event tree analysis~(ETA),
and layer of protection analysis~(LOPA) allow the systematic
identification and investigation of dangerous events and the
derivation and assessment of interventions.
\emph{Fault tree analysis}~(FTA) and \emph{failure mode effects
  analysis}~(FMEA) are widely used in structured analysis, assessment,
and reduction of system or process failures.  FTA follows a deductive
scheme of causal reasoning, FMEA an inductive one.

There are many variations, extensions, %
combinations of these techniques, integrated into assurance
approaches~\citep{McDermid1994-Supportsafetycases}, enriched with
system models~(e.g.\xspace the use of UML in
HazOp~\citep{Guiochet2010-ExperienceModelBased}),
coming along with intuitive visual languages, and tailored for
specific applications or stages in the system life
cycle~\citep{Ericson2015-HazardAnalysisTechniques}.
All these techniques serve the analysis of undesired events, their
causes and consequences, with the aim of reducing risk.
\begin{inparaitem}
\item \emph{Reliability techniques} help reducing the risk of
  \emph{failures}~(i.e.,\xspace reduce fault-proneness and increase
  fault-tolerance) of mechanical and mechatronic
  systems~\citep{%
    Birolini2017-ReliabilityEngineering}
\item \emph{Safety techniques} focus on reducing the risk of
  \emph{dangerous failures} of control systems and
  software~\citep{Leveson1995-SafewareSystemSafety}.
\end{inparaitem}
Both directions differ, usually overlap, and are embedded into the
context of dependability engineering.

\paragraph{Analysis and Reduction of Random Failures.}
\label{sec:relwork:randfail}

The body of scientific literature on failure analysis is overwhelming.
The following overview of techniques focuses on a small fraction of
model-based and formal techniques for failure analysis at design-time.
Before summarising several techniques, we will have a brief look onto
the formal concepts of FTA because it is one of the most widely
practised versatile techniques with many powerful
extensions~\citep{Ericson2015-HazardAnalysisTechniques}.

A fault tree is a causal model of a system relating an undesired
\emph{top-level event} $\mathit{e_{TL}}$ with a set $B$ of \emph{basic
  events} using various kinds of \emph{gates}~(e.g.\xspace \texttt{AND},
\texttt{OR}, \texttt{NOT}) %
connecting these events.  A \emph{minimum cut set}~(MCS) is a minimum
set of events required to occur to activate $e_{TL}$.  \emph{Dynamic
  fault trees} are an important extension of fault trees, allowing to
model the order of events and, thus, leading to \emph{minimum cut
  sequences}.  A minimum cut sequence describes a minimum set of
events that have to occur in this sequence to activate $e_{TL}$.
Fault trees can also be expressed by connecting specific sets
$MCSets\subseteq2^{B}$ in the disjunctive normal form in the
antecedent and with $e_{TL}$ in the consequent of the following
implication:
\[
  \big( \bigvee_{S\in \mathit{MCSets}}\bigwedge_{e_B \in S} e_B \big) \Rightarrow \mathit{e_{TL}}
\]
where $MCSets$ denotes all MCSs and $e_B$
stands for a \emph{basic event}. %

FTA and FMEA are regularly used in combination and have been given
formal semantics also based on probabilistic
models~\citep{Kumamoto2007-Satisfyingsafetygoals}.  For
quantitative analysis, many approaches use \textsc{Markov}
models~\citep{Dehlinger2008-DynamicEventFaultTree}, or at a higher
level of abstraction, variants of stochastic \textsc{Petri}
nets~\citep{%
  Papadopoulos2011}. %
For dynamic fault trees,
\citet{Kabir2018-UncertaintyAwareDynamic} show how fuzzy numbers
can be used %
to integrate qualitative expert opinions on unknown failure rates into
a generalised stochastic \textsc{Petri} net for quantitative reliability
evaluation.

Complex dynamic fault trees can also be converted into probabilistic
automata for the checking of stochastic temporal
properties~\citep{%
  Boudali2007-Dynamicfaulttree}
and efficiently synthesise failure
rates~\citep{Volk2016-AdvancingDynamicFault}.
Given a model of the system under consideration, fault trees can be
synthesised from failure
annotations~\citep{Unanue2018-ExplicitModellingTreatment} or from
counterexamples generated by model
checkers~\citep{%
  Leitner-Fischer2013-ProbabilisticFaultTree}.

Bow tie analysis~(BTA) combines FTA and FMEA or ETA.
\citet{Denney2017-ModelDrivenDevelopment} developed a tool for
modelling and managing many and large handcrafted and interrelated
bow tie diagrams and for quantitative assessment.  The authors
explain how bow ties can guide the construction of assurance cases.
Using a hierarchical control model of the process, system-theoretic
process analysis~(STPA) \citep{Leveson2012-EngineeringSaferWorld} aims
at holistic safety assessment.  STPA shares with BTA and HazOp its aim
to bridge the gap between failure and accident analysis.  STPA's
conceptual framework~(called STAMP)
\citep{Leveson2004-newaccidentmodel} promotes the light-weight and
abstract modelling of control hierarchies.  Variants of STPA have been
proposed specifically for the analysis of accidents~(CAST)
\citep{Stringfellow2010-AccidentAnalysisHazard}.
Finally, why-because analysis~(WBA)
\citep[Ch.~20]{Ladkin2001-CausalSystemAnalysis} provides a
conceptual, formal, and graphical framework for cross-technology and
cross-disciplinary root cause identification. Like the majority of the
aforementioned techniques, WBA is primarily deductive.

Depending on the possibilities, failures can be reduced in multiple
ways.  Accordingly, a variety of paradigms is available.  Irreducible
undesired events~(e.g.\xspace equipment failures, maloperation), once
identified and assessed, can be drastically reduced by fault-tolerant
design~(e.g.\xspace\citep{Littlewood2011-ReasoningReliabilityDiverse}) or
active or passive measures~(e.g.\xspace safety functions, physical
separation).  For example,\xspace \citet{Michalos2015-DesignConsiderationsSafe}
provide an overview of safety measures built into and around robots
collaborating with humans in manufacturing automation.

\paragraph{Analysis and Reduction of Systematic Failures.}
\label{sec:relwork:systfail}

As opposed to random failures, systematic failures suggest
further means of reduction or even avoidance, applicable much earlier
in the engineering process.

Following quality management paradigms in
manufacturing, %
engineering can also be viewed as a stochastic
process~\citep{%
Littlewood1991-Softwarereliabilitymodelling}.  Stochastic
variables are formed by the circumstances~(e.g.\xspace location, time,
technique) under which faults are introduced and detected.
Such faults are typically systematic~(also called development
failures~\citep{Avizienis2004-Basicconceptstaxonomy}) because they
have an impact on the specification and design leading to operational
failures that can usually be fully reproduced as opposed to random
failures.

To reduce early root causes, particularly, for systematic failures,
\emph{formal methods} have been proposed.  The main argument for the
use of such mathematically founded techniques is their power in detecting
errors and inconsistencies in requirements, algorithms, and
designs~(e.g.\xspace invariant violation, deadlock, starvation
\citep{Roscoe2010, Schneider1999-ConcurrentRealTime}) early before
the system is built or put in operation.
In the context of combining different paradigms, fault trees~(see
page~\pageref{sec:relwork:randfail}) are often used as a feedback tool
for incremental design improvement.  For example,\xspace
\citet{Hansen1998-SafetyAnalysisSoftware} discuss how formal fault
trees can be used to derive safety requirements to validate the
software design during its step-wise development.
Because the introduction of formal methods is difficult and
proposed methods occurred to be inadequate,
\citet{Bowen1993-Safetycriticalsystems} investigated the
applicability of formal methods in safety-critical software
engineering.  Most importantly, the authors pointed out a typical 
problem of formal methods, the difficulty of safety provisions and
measurements when concerned with a degree of safety or safety as risk
instead of being an invariant not to violate.

\paragraph{Analysis and Mitigation at Run-time.}
\label{sec:backg:monitors}

One striking advantage of using model-based techniques, particularly,
formal methods, is the possibility of using property specifications and
models at run-time.
For example,\xspace in run-time verification or
monitoring~\citep{%
  Leucker2009-briefaccountruntime}, properties are checked during
system operation by recording observation traces~(e.g.\xspace values of
observed variables) and checking these traces for violations~(safety)
or for acceptance~(co-safety).  The checking task is performed by
independent components, sometimes called watchdogs, safety monitors,
or policing
functions~\citep{Bogdiukiewicz2017-FormalDevelopmentPolicing}.
Monitoring can be used to derive probabilistic statements about system
health at any time during operation, for example,\xspace using \textsc{Bayes}ian
networks~\citep{Iamsumang2018-MonitoringLearningAlgorithms}.

\subsection{Formal Preliminaries}
\label{sec:formal-preliminaries}

This section prepares the preliminaries for the formalisation in the
later sections.

\paragraph{Processes.}
\label{sec:kripke-struct-risk}

We use communicating sequential processes~(CSP)~\citep{Hoare1985}
for system modelling.
Let $\Sigma$ be the set of all~(labels for) \emph{events} called the
\emph{alphabet} of a process.  We distinguish two special events:
$\checkmark$ signifies the \emph{successful termination} of a process
and $\tau$ signifies the \emph{invisible event} resulting
from abstraction~(i.e.,\xspace \emph{hiding} or constrained observation).
We require
$\checkmark,\tau \not\in \Sigma$ and define
$\Sigma^{\checkmark} = \Sigma \cup \{\checkmark\}$,
$\Sigma^{\tau} = \Sigma \cup \{\tau\}$, and
$\Sigma^{\checkmark,\tau} = \Sigma^{\checkmark} \cup
\Sigma^{\tau}$.

\begin{definition}[Process]
  \label{def:process}
  A \emph{process} is an expression of the following form:
  \[
    P,Q ::=
    a \operatorname{\rightarrow} P
    \mid \prCh{A} P
    \mid P \sqcap Q
    \mid P \operatorname{\square} Q
    \mid P \genpar[A] Q
    \mid P ; Q
    \mid P \setminus A
    \mid \mathit{SKIP}
    \mid \mathit{STOP}
  \]
  where $a \in \Sigma$ and $A \subseteq \Sigma$.  Let $\mathcal{P}$ be
  the set of all~(labels for) \emph{processes}~(or, equivalently, their
  control states) %
  with $\mathit{SKIP}, \mathit{STOP} \in \mathcal{P}$.
\end{definition}

The syntax consists of the following constructs:
\begin{inparaitem}
\item \emph{event prefix}~($a \operatorname{\rightarrow} P$), 
\item \emph{prefix choice}~($\prCh{A} P$), 
\item \emph{internal or nondeterministic choice}~($P \sqcap Q$), 
\item \emph{general choice}~($P \operatorname{\square} Q$),
\item \emph{generalised parallel composition}~($P \genpar[A] Q$),
\item \emph{sequential composition}~($P ; Q$),
\item \emph{event hiding}~($P \setminus A$),\footnote{To avoid confusion
    with the hiding operator $\setminus$, we will use $\backprime$ for
    set subtraction in CSP expressions.}
\item \emph{termination}~($\mathit{SKIP}$), and
\item \emph{deadlock}~($\mathit{STOP}$).
\end{inparaitem}

We attach meaning to each expression $P\in\mathcal{P}$ according to
\Cref{def:process} by providing a recursive scheme denoting the
behaviour of $P$.  Among several ways of doing this, we focus on the
\emph{traces model} $\mathcal{T}$.  Traces are finite sequences over
$\Sigma^{\checkmark}$ abstracting from details~(e.g.\xspace internal state)
of $P$'s executions or
behaviours~\citep{Schneider1999-ConcurrentRealTime}: ``a trace is
a record of the visible events of an execution.''  In $\mathcal{T}$, we
use a function $\traces{P}$ to obtain the \emph{trace semantics} of $P$.
For example,\xspace we have $\traces{\mathit{STOP}} = \{\trace\}$.  The event $\checkmark$
represents the counterpart of $\mathit{SKIP}$ in $\mathcal{T}$ as denoted by
$\traces{\mathit{SKIP}} = \{\trace,\trace[\checkmark]\}$.
$initials(P) = \{a \mid \trace[a] \in \traces{P}\}$ denotes the set of
all \emph{initial} events of $P$.
If we can establish the relation $traces(P) \subseteq traces(Q)$ for
two processes $P$ and $Q$, we write $P \refinedBy Q$ and say that $Q$
\emph{refines} $P$~(or $P$ is refined by $Q$).

With \Cref{def:process}, we can model observable behaviour of systems
and refer to distinct portions of such behaviour using control state
labels.
For a comprehensive account of CSP and a hierarchy of CSP models, the
inclined reader may consult~\citep{%
  Roscoe2010}.

\paragraph{Transition Systems.}
\label{sec:transition-systems}

For the investigation of the operational semantics of both risk models
and CSP, labelled transition systems~(LTS,
\cite{Baier2008-PrinciplesModelChecking}) will be defined along
with notational conventions.

\begin{definition}[Labelled Transition System, LTS]
  \label{def:risk:lts}
  A LTS is a tuple $\mathfrak{S} = (S,E,\rightarrow,S_0)$ with a set
  $S$ of \emph{states}, a set $E$ of \emph{events}, a relation
  $\rightarrow\;\subseteq S \times E \times S$ representing
  transitions between these states when engaging in these events, and a
  set $S_0$ of \emph{initial states}.
\end{definition}
$runs(s,\mathfrak{S})$ denotes the runs~(i.e.,\xspace state/event traces)
observable from the process modelled by $\mathfrak{S}$ in state $s$.
We write $\tran{s}{e}{s'}$ if $(s,e,s') \in\;\rightarrow$,
$\tran{s}{e}{}$ if $\exists s'\in S \colon (s,e,s') \in\;\rightarrow$,
$\tran{s}{}{}$ if
$\exists e \in E , s'\in S \colon (s,e,s') \in\;\rightarrow$,
$\ntran[]{s}{e}{}$ if
$\not\exists s'\in S \colon (s,e,s')\in\;\rightarrow$, and
$\ntran[]{s}{}{}$ if
$\not\exists e \in E ,s'\in S \colon (s,e,s')\in\;\rightarrow$.
The omission of event labels and initial states leads to the more
general form $(S,\rightarrow)$ with
$\rightarrow\;\subseteq S \times S$, called transition system~(TS).

\paragraph{Temporal Logic.}
\label{sec:temporal-logic}

For the investigation of constraints on risk models, we employ
linear-time temporal
logic~\citep[TL]{Manna1991-TemporalLogicReactive}.  For TL
formulae $\phi$, $\psi$, a run $\rho\in runs(s,\mathfrak{S})$, and the
truth constants $\textsc{t}$ and $\textsc{f}$,
\begin{inparaitem}
\item the operator $\circ\phi$ expresses that $\phi$ holds of the
  \emph{next state} of $\rho$ and
\item $\phi\mathrel{\mathsf{U}}\psi$ expresses that $\phi$ holds of \emph{every}
  state \emph{until} $\psi$ holds of a state, with $\psi$ required to
  eventually hold.
\item For convenience, $\square\phi = \phi \mathrel{\mathsf{U}} \textsc{f}$
  denotes that $\phi$ holds of every state of a \emph{whole} run,
\item $\Diamond\phi = \textsc{t} \mathrel{\mathsf{U}} \phi$ expresses that
  $\phi$ holds \emph{eventually} or for some future state, and
\item $\phi\mathrel{\mathsf{W}}\psi = \phi\mathrel{\mathsf{U}}\psi \lor
  \square \phi$ allows $\psi$ to never actually hold. %
\item For a last state of any run prefix, $\blacklozenge\phi$ denotes that
  $\phi$ has held \emph{before} or in the \emph{past} represented by
  this prefix.
\end{inparaitem}
TL formulas can be interpreted for \emph{events} in the same way.  In
timed extensions of
TL~\citep{Koymans1990-Specifyingrealtime},
we allow time constraints of the form ``${\sim}t$''
to be attached to some of these operators, with
$\sim \;\in \{<,>,\leq,\geq\}$. For example,\xspace $\Diamond_{<\,t}\phi$
expresses that $\phi$ holds for \emph{some} future state \emph{before}
$t$ time units will have elapsed.
Events in runs can be treated in a similar way and the satisfaction
relation %
can be extended to sets of runs of $\mathfrak{S}$.  A comprehensive
treatment is provided in~\citep{Baier2008-PrinciplesModelChecking}.

\section{Risk Elements}
\label{sec:risk:model}

This section introduces a model to \emph{capture beliefs} about risk
and risk causality based on a process $P$~(\Cref{def:process}).  It is
not unusual to consider one part $S$ of $P$ as the ``system under
consideration'' interacting with another part $E$ of $P$ called
``environment'' or ``context''.  We may consider two settings:
$P = S \genpar[A] E$ with a set of shared events $A$, or $P = E(S)$
where $E$ is a term using $S$.  

Based on the discussion in \Cref{sec:risk}, we view \emph{risk} as the
\emph{possibility%
  \footnote{For the sake of simplicity of the presented framework, we omit
    probabilistic aspects for the time being.}  of undesired states}
\begin{inparaenum}[(i)]
\item \emph{reachable} by $P$,
\item \emph{finitely causal}\footnote{$\traces{P}$ contains finite
    traces reaching such states such that causes can be represented by
  well-founded sets.}, and
\item \emph{not entirely avoidable} in $P$.
\end{inparaenum}
Undesired states both result from \emph{undesired events} and may
cause or at least increase the risk of further undesired events,
hence, forming causal chains of events.  As highlighted in
\Cref{sec:introduction}, in our risk models, we concentrate on
\emph{undesired fractions of uncertain outcomes of process actions},
that is,\xspace observable events changing the level of risk from process state to
process state.  Let us consider further examples before we formalise
these concepts.

\begin{example}[Road Traffic and Brakes]
  \label{exa:risk:roadtraffic}
  \emph{Collisions} are examples of undesired events, accidents
  resulting in undesired states entailing human injury and damaged
  property.  Even \emph{near-collisions} are undesired events, by
  definition posing the risk of actual collisions.  By backward
  reasoning, for example,\xspace an observable loss of a car's braking function
  constitutes an undesired event leading to an undesired state of the
  car because of operating without functioning brakes.  Clearly, such
  a vehicle is in a riskier state than with functioning brakes.
\end{example}

\Cref{exa:risk:roadtraffic} highlights why the state bi-partition in
\Cref{fig:abslevels} from the viewpoint of a single risk factor is too
coarse.  Failures of a brake or an anti-lock braking system are not
repairable at run-time, in other words, \emph{direct mitigation} is
not possible.  Hence, \Cref{fig:riskfactor} introduces a third
partition as described below.  Moreover, a failure of the brakes
\emph{restricts} %
the direct mitigation of near-collisions~(\Cref{exa:evaloperincid}) by
braking.  This example also shows how two risk factors are related.

\begin{example}[Autonomous Vehicles and Brakes]
  \label{exa:risk:avandbrakes}
  There is a difference between \emph{human-operated cars}, where
  human operators might be \emph{aware} of missing brakes, and
  \emph{autonomous vehicles}~(AVs) where \emph{being aware} is left to
  the automation.  In the former type of cars, alert %
  human operators will try to react.  Independent of human
  capabilities, what matters is that they get aware and are given the
  chance to take responsibility of emergency control.  However, AVs
  are left alone in such a situation, unlikely will their vendors
  manage to legitimately push away such responsibility.  We might
  expect from AVs to implement at least as much responsibility as
  society would have expected from qualified human drivers in the
  corresponding driving situations.
\end{example}

\Cref{exa:risk:avandbrakes} motivates what safety engineers can do to
equip autonomous machines with the ability to run highly specific
\emph{mitigation mechanisms} in risky states, to develop beliefs about
past operations and to predict risk in future operations of a machine
and certify that such mechanisms actually improve safety.

\subsection{Risk Factors}
\label{sec:risk:factors}

\emph{Risk factors} form the basic
elements %
of the approach to risk modelling as discussed in the following.  We
define the notion of a risk factor using an LTS, describe its
properties and meaning, provide a translation into CSP, and discuss an
algebra of risk factors.

\begin{definition}[Risk Factor]
  \label{def:riskfactor}
  Let $(\mathit{Ph}, \Sigma^{\rfct}, \rightarrow)$ be a LTS according to
  \Cref{def:risk:lts}.  Extending this LTS, a \emph{risk factor} is a
  tuple
  $\rfct = (\mathit{Ph}, \Sigma^{\rfct}, \rightarrow, \preceq_{\mathsf{\rfct}}, s)$ with
  \begin{itemize}
  \item a set $\mathit{Ph}$ of \emph{phases} of $\rfct$,
  \item a set
    $\Sigma^{\rfct}\subset 2^{\Sigma^{\tau}}
    \setminus \{\varnothing\}$ specifying \emph{significant events} for
    $\rfct$,
  \item a labelled transition relation
    $\rightarrow\; \subseteq \mathit{Ph} \times \Sigma^{\rfct} \times
    \mathit{Ph}$, 
  \item a partial order
    $\preceq_{\mathsf{\rfct}} \;\subseteq \mathit{Ph}^2$, and
  \item a pair $s \in \mathbb{R}^2_+$ with $s^{(1)} \leq s^{(2)}$
    for the \emph{severity of the least and worst expected
      consequences} of $\rfct$.\footnote{By usual convention, for an
      ordered $n$-tuple $t$ and $i \in [1..n]$, we write $t^{(i)}$ to
      refer to the value of the $i$-th element of $t$.  Furthermore,
      if $t$ has a uniquely named element $e$, we write $t.e$ to refer
      to the value of $e$ in $t$.}
  \end{itemize}
  We call $c$ the \emph{severity (interval)} of $\rfct$.  Let $\mathbb{F}$ be the
  set of all risk factors in the remainder of this work.
\end{definition}

We only consider risk factors with finite
$\mathit{Ph}$ and $\Sigma^{\rfct}$.  Furthermore, we
regard risk factors $\rfct$ with $\rightarrow$ according to
\Cref{fig:riskfactor} with
$\mathit{Ph} = \{\inact,\activ,\mitig\}$ for the phases
\emph{inactive}~($\inact$, typically, the initial and desired phase
of $\rfct$), %
\emph{active}~($\activ$), and \emph{mitigated}~($\mitig$),
where $\preceq_{\mathsf{\rfct}}$ is at least\footnote{Some applications might
  give rise to a \emph{linear} $\preceq_{\mathsf{\rfct}}$ by adding at most one
  out of
  $\{ (\inact,\mitig), (\mitig,\inact) \}$.}
the reflexive transitive closure of
$\{ (\activ,\inact), (\activ,\mitig) \}$, and with
$\Sigma^{\rfct} = \{
e^{\rfct},
\overline{e}^{\rfct},
m^{\rfct},
m_d^{\rfct},
m_r^{\rfct},
o_n^{\rfct}, %
o_m^{\rfct}, %
o_e^{\rfct} \}$. %
The labels in \Cref{fig:riskfactor} indicate the meanings of these
events.  \Cref{tab:risk:eventsets} describes how these events can be
used for modelling risk factors.

\begin{figure}
  \subfloat[Symbolic transition system for $\rfct$]{
\begin{tikzpicture}
  [->,>=stealth,thick,
  node distance=10em,
  scale=.7,every node/.style={transform shape},
  shorten >=1pt,auto,%
  state/.style={circle,fill=LightBlue!50}]

  \node[state] (0) {$\inact$}; %
  \node[state] (c) [above right of=0] {$\mitig$}; %
  \node[state] (1) [below right of=c] {$\activ$};
  \node (init) at ($(0)+(0,-1)$) {}; 

  \path
  (init) edge[gray] (0)
  (0) edge[red]  node[align=center] {$e^{\rfct}$ \dots endangerment} (1) 
  (0) edge[loop left,gray] node[align=right]
  {$o_n^{\rfct}$\\ nominal\\operation\dots} (0)
  (1) edge[bend right,swap,DarkGreen] node[align=center]
  {$m^{\rfct}$ \dots mitigation} (c)
  (c) edge[bend right,DarkGreen,dashed] node[swap,align=center]
  {recovery\dots $\;m_r^{\rfct}$} (0) %
  (1) edge[bend left,DarkGreen,near start] node[align=center] {$m_d^{\rfct}$
    \dots direct mitigation} (0) %
  (c) edge[gray] node[near start,swap,align=right]
  {endanger-\\ment \dots $\;\overline{e}^{\rfct}$} (1)
  (1) edge[loop right,gray] node[align=left]
  {$o_e^{\rfct}$\\ \dots endangered\\operation} (1)
  (c) edge[loop above,gray] node[align=center]
  {$o_m^{\rfct}$ \dots mitigated operation} (c);
\end{tikzpicture}

     \label{fig:riskfactor:graph}
  }
  \subfloat[Types of events for abstracting from $P$ to $\rfct$]{
    \small
    \begin{tabular}[b]{lp{5cm}r}
      \toprule
      \textbf{Symbol}
      & \textbf{\dots represents events} \dots
      & \textbf{Ex.}
      \\\midrule
      $e^{\rfct}$ 
      & leading to $\activ$, $P$ in endangered oper.
      & \ref{exa:strredrf} 
      \\
      $\overline{e}^{\rfct}$
      & leading to $\activ$, $P$ in endangered oper.
      & \ref{exa:strredrf} 
      \\
      $m^{\rfct}$
      & leading to $\mitig$, $P$'s mitigated
        operation
      & \ref{exa:strredrf}
      \\
      $m_d^{\rfct}$
      & leading to $\inact$, $P$'s nominal operation
      & \ref{exa:dirredrf}
      \\
      $m_r^{\rfct}$
      & recovering to nominal operation of $P$
      & \ref{exa:indredrf}
      \\
      $o_n^{\rfct}$
      & not leaving $\inact$
      \\
      $o_e^{\rfct}$
      & not leaving $\activ$
      & \ref{exa:strredrf}
      \\
      $o_m^{\rfct}$
      & not leaving $\mitig$
      \\\bottomrule
    \end{tabular}
    \label{tab:risk:eventsets}
  }
  \caption{Risk factor $\rfct$
    \label{fig:riskfactor}}
\end{figure}
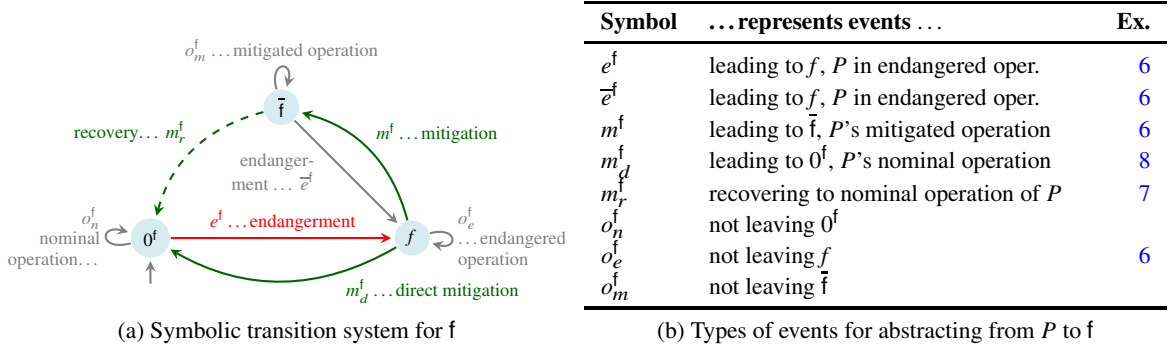

The three state preserving events can complement the endangerment and
mitigation events.
By definition, if $e = \varnothing$ then
$(p,e,p')\not\in\;\rightarrow$ for any $p,p'\in \mathit{Ph}$.  We can
now distinguish further kinds of risk factors:
\begin{itemize}
\item If $m^{\rfct} \cup m_d^{\rfct} = \varnothing$ then we call
  $\rfct$ \emph{final}, otherwise \emph{reducible}.
\item For any reducible $\rfct$ with $m^{\rfct} \neq \varnothing$, if
  $\overline{e}^{\rfct} \subset e^{\rfct}$ then we call $\rfct$
  \emph{strongly reducible}. %
\item If $m_d^{\rfct} = \varnothing$ then we call $\rfct$
  \emph{indirectly reducible}.
\end{itemize}
The type of risk factor described in \Cref{fig:riskfactor:graph} is
\emph{minimal} inasmuch as it comprises the minimal set of elements of
\emph{generic risk atoms}.  However, phases other than the ones
shown can be distinguished in specific applications.

\paragraph{Abstractions underlying Events.}

Events in the CSP interpretation are atomic
observations~\cite[Ch.~1.5]{Roscoe2010}.  The initiation and
termination of events representing complex enduring real-world
phenomena are to be viewed as non-separable aspects of these events.
Consequently, much care is necessary when making assumptions about the
atomicity of such events interfering with other events.  Hence, it is
reasonable to use the types of events listed for risk factors to model
the \emph{initiation, termination, or other significant events} of
  the corresponding sub-processes~(i.e.,\xspace endangerment and mitigation
  processes) in the process $P$.

\begin{example}[Final Risk Factors]
  \label{exa:risk:finalriskfactor}
  Road accidents as well as nuclear power-plant accidents form courses
  of events.  Severe human injury or loss and machine, environmental,
  and property damage typically happen during such accidents.  If
  required, we can model such injury or damage as \emph{final} risk
  factors and, thus, can stop to discuss possible mitigations.  This
  way, final risk factors define the \emph{scope} of a risk model.

  Viewing \emph{damage or injury} as risk factors allows their
  treatment within the same framework.  We might later introduce
  mitigations and convert a final into a reducible risk factor.
  Airbags as a mitigation of certain types of human injury for certain
  types of collisions represent a historical example.
\end{example}

\begin{example}[Strongly Reducible Risk Factors]
  \label{exa:strredrf}
  We instantiate the LTS pattern given in \Cref{fig:riskfactor}:
  Consider the \textbf{event} $\acend{db}$~(instance of $\acend{f}$)
  that \emph{a car's braking function degrades} resulting in an
  \textbf{operational state} $\activ[db]$ of the car where any further
  use of its brakes~($\acopr[e]{db}$) will likely differ from the
  expectation of a human operator when trying to reduce speed in a
  typical driving situation.  One conservative
  \textbf{mitigation}~($\acmit{f}$) would be to drive by and halt as
  safely as possible and, thus, reach a \textbf{state} $\mitig[db]$
  from which only a strict subset~(i.e.,\xspace
  $\acendr{db} \subset \acend{db}$) of the original endangerment event
  can be observed.  Hence, we call $\rfct[db]$ a strongly reducible
  risk factor.
\end{example}

\begin{example}[Indirectly Reducible Risk Factors]
  \label{exa:indredrf}
  If our application requires an intermediate stable state
  $\mitig$ for the mitigation of a risk factor before
  returning to the inactive phase $\inact$, we speak of indirectly
  reducible risk factors.

  A \emph{leaking or damaged battery} of an AV would be such a case as
  well as an aircraft \emph{running out of fuel}.  If we do not want to
  consider ways to fuel such machines during operation but by reaching
  what is typically called a ``safe state,'' we might model such
  situations by indirectly reducible risk factors.  The
  \emph{mitigated} phase would represent the ``safe state'' with respect to\xspace this
  risk factor.  This way, our model captures how to reach this phase.
  In our example, this can happen by the atomic events
  \emph{successfully halting at the next car repair shop} and
  \emph{successful accomplishment of an emergency landing},
  respectively.  From this phase, we can recover~($\acmit[r]{f}$) to
  the inactive phase $\inact$.
\end{example}

\begin{example}[Directly Reducible Risk Factors]
  \label{exa:dirredrf}
  Driving \emph{too close to a front vehicle} is a risk factor that in
  many situations can be dealt with by braking correspondingly and,
  thus, resulting in a state where this risk factor is inactive again.
  The described braking event can be accounted as an event in
  $\acmit[d]{f}$.  We call this a \emph{direct mitigation}. %
\end{example}

\paragraph{Risk Factors in CSP.}
\label{sec:risk:riskfactors-in-csp}

Given $\ntran[]{p}{\tau}{}$ for any $p\in\mathit{Ph}$, the
risk factor $\rfct$ from \Cref{fig:riskfactor} can be represented as a
sequential, mutually recursive CSP process
$\mathfrak{R}_{\rfct}$~(\Cref{def:process}) as follows:
\begin{align*}
  \mathfrak{R}_{\rfct}
  &= \inact
    \tag{init}\\
  \inact
  &= \prCh{\acend{f}} \activ
    \operatorname{\square} \prCh{\acopr[n]{f}} \inact
    \tag{inactive}\\
  \activ
  &= \prCh{\acmit[d]{f}} \inact
    \operatorname{\square} \prCh{\acopr[e]{f}} \activ
    \operatorname{\square} \prCh{\acmit{f}} \mitig
    \tag{active}\\
  \mitig
  &= \prCh{\acendr{f}} \activ
    \operatorname{\square} \prCh{\acmit[r]{f}} \inact
    \operatorname{\square} \prCh{\acopr[m]{f}} \mitig
    \tag{mitigated}
\end{align*}

If
$\acopr[n]{f} \subseteq \Sigma\setminus \acend{f}, \acopr[m]{f}
\subseteq \Sigma\setminus (\acmit[r]{f} \cup \acendr{f})$,
and
$\acopr[e]{f} \subseteq \Sigma\setminus (\acmit{f} \cup \acmit[d]{f})$
then the general choice gets an external choice and $\mathfrak{R}_{\rfct}$
is \emph{deterministic}, otherwise $\mathfrak{R}_{\rfct}$ is nondeterministic.
This mapping enables an algebraic treatment of risk factors in CSP.
Later, in \Cref{sec:risk-structures}, we will use a map
$[\hspace{-.15em}[\cdot]\hspace{-.15em}]_{r}$ that establishes the equivalence
$[\hspace{-.15em}[\inact]\hspace{-.15em}]_{r} = (\mathit{Ph}, \Sigma^{\rfct},
\rightarrow,\{\inact\})$ for a factor $\rfct$ according to
\Cref{def:riskfactor} and initialised with the phase $\inact$.
Note the use of $\rfct$ as a symbol for the risk factor as a
transition system~(\Cref{fig:riskfactor:graph}) and the use of
$\activ$ for the CSP process that models this risk factor in its
active phase.  Different fonts signify the semantic difference.

\begin{remark}
  Risk factors can be used to model \emph{risky fractions of or
    propositions about} a process %
  and its behaviour.  For example,\xspace final risk factors can be used to model, for example,\xspace
  \emph{permanent and off-line repairable faults}, and reducible risk
  factors serve the modelling of, for example,\xspace \emph{transient and on-line
    repairable faults}.  This model only allows to talk of risks
  identified as risk factors and cannot be used to reason about
  ``absolute safety.''  This epistemic limit is inherent to~(risk)
  modelling and can only be dealt with from the outside of the
  framework.
\end{remark}

\begin{remark}
  The notions of \emph{systematic and random
    faults}~\citep{%
    Birolini2017-ReliabilityEngineering}
  can be represented as follows:
  A \emph{systematic fault} can be seen as an observable undesired
  event whose preconditions or causes~(i.e.,\xspace $\inact,\acend{f}$) can
  be predicted, reconstructed, reproduced, or otherwise deduced,
  identified, and sufficiently determined.  Hence, each~(class of)
  systematic fault(s) can be associated with a deterministic risk
  factor.

  A \emph{random fault} can be seen as an observable undesired event
  whose preconditions or causes are only partially known or even
  unknown.  One way to represent this lack of knowledge by a risk
  factor is to use nondeterminism:\footnote{Another way, not pursued
    in this article, would be to provide a probabilistic extension of
    risk factors.}  From $\inact$, we allow
  $rnd = \acopr[n]{f} \cap \acend{f} \neq \varnothing$.  $rnd$
  represents potential but incomplete causes of $\rfct$.  Note that
  this choice makes sense inasmuch as
  $\acopr[n]{f} \supseteq \acend{f}$ denotes that we know $\rfct$ but
  the least possible, namely nothing, about its causes.  Having
  observers\footnote{The usage of a risk factor as a monitor automaton
    is further discussed in \Cref{sec:notes-monitoring}.}  for events
  and phases, the risk factor would form a model of a process $P$
  where an observation of an event of $P$ in $rnd$ is sometimes
  followed by an observation of the phase $\inact$ and sometimes by
  an observation of the phase $\activ$.
  $\rfct$ in \Cref{fig:riskfactor:ndet} is a generalisation of $\rfct$
  in \Cref{fig:riskfactor:graph}~(i.e.,\xspace
  $\rfct_{(\ref{fig:riskfactor:ndet})}\refinedBy\rfct_{(\ref{fig:riskfactor:graph})}$).
  Each observable event of $P$ that belongs to $rnd_0$~($rnd$ for
  $\rfct$, $\overline{rnd}$ for $\mitig$) is leading to an
  anonymous phase~($\bullet$) succeeded by a $\tau$ event
  representing
  internal choice. %
  For mitigation events $\acmit[d]{f}$ and $\acmit{f}$ to be observed
  from the active phase $\activ$, uncertainty is modelled by the set
  $rnd$, again to be followed by either of the three phases of the
  risk factor depending on what can be observed in $P$.
\end{remark}

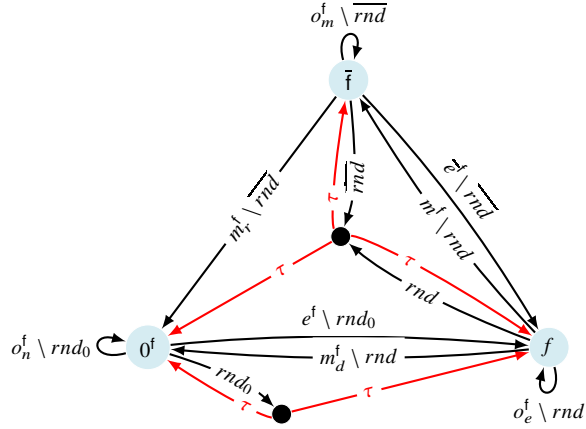
\begin{figure}
  \centering

\begin{tikzpicture}[>=latex,line join=bevel,scale=0.7, every node/.style={scale=.8}]
\begin{scope}
  \pgfsetstrokecolor{black}
  \definecolor{strokecol}{rgb}{1.0,1.0,1.0};
  \pgfsetstrokecolor{strokecol}
  \definecolor{fillcol}{rgb}{1.0,1.0,1.0};
  \pgfsetfillcolor{fillcol}
\end{scope}
\begin{scope}
  \pgfsetstrokecolor{black}
  \definecolor{strokecol}{rgb}{1.0,1.0,1.0};
  \pgfsetstrokecolor{strokecol}
  \definecolor{fillcol}{rgb}{1.0,1.0,1.0};
  \pgfsetfillcolor{fillcol}
\end{scope}
  \node (inactive) at (148.56bp,45.174bp) [draw=white,ellipse,circle,fill=LightBlue!50] {$\inact$};
  \node (active) at (364.56bp,45.174bp) [draw=white,ellipse,circle,fill=LightBlue!50] {$\activ$};
  \node (mitigated) at (256.56bp,189.17bp) [draw=white,ellipse,circle,fill=LightBlue!50] {$\mitig$};
  \node (a1) at (220.56bp,9.1736bp) [draw=white,circle, fill,fill=black] {};
  \node (a2) at (252.39bp,104.84bp) [draw=white,circle, fill,fill=black] {};
  \draw [->,arrows=-latex,thick] (inactive) ..controls (217.0bp,52.357bp) and (281.93bp,52.534bp)  .. node[midway,above,fill=white!100] {$e^{\rfct} \setminus rnd_0$} (active);
  \draw [->,arrows=-latex,thick] (active) ..controls (323.25bp,88.541bp) and (296.14bp,124.33bp)  .. node[sloped,fill=white!100] {$m^{\rfct} \setminus rnd$} (mitigated);
  \draw [red,->,arrows=-latex,thick] (a1) ..controls (233.61bp,12.436bp) and (286.98bp,25.78bp)  .. node[midway,sloped,fill=white!100] {$\tau$} (active);
  \draw [->,arrows=-latex,thick] (active) ..controls (310.27bp,65.57bp) and (276.12bp,84.401bp)  .. node[sloped,fill=white!100] {$rnd$} (a2);
  \draw [->,arrows=-latex,thick] (mitigated) ..controls (261.34bp,142.41bp) and (259.93bp,126.17bp)  .. node[sloped,fill=white!100] {$\overline{rnd}$} (a2);
  \draw [->,arrows=-latex,thick] (inactive) ..controls (191.08bp,31.857bp) and (204.94bp,24.253bp)  .. node[sloped,fill=white!100] {$rnd_0$} (a1);
  \draw [->,arrows=-latex,thick] (active) ..controls (296.12bp,37.991bp) and (231.18bp,37.813bp)  .. node[sloped,fill=white!100] {$m_d^{\rfct} \setminus rnd$} (inactive);
  \draw [->,arrows=-latex,thick] (mitigated) ..controls (219.01bp,139.11bp) and (192.26bp,103.45bp)  .. node[sloped,fill=white!100] {$m_r^{\rfct} \setminus \overline{rnd}$} (inactive);
  \draw [->,arrows=-latex,thick] (mitigated) ..controls (299.59bp,143.59bp) and (326.76bp,107.6bp)  .. node[above,sloped] {$\overline{e}^{\rfct} \setminus \overline{rnd}$} (active);
  \draw [->,arrows=-latex,thick] (active) to[loop below] node[below] {$o_e^{\rfct} \setminus rnd$} (active);
  \draw [->,arrows=-latex,thick] (mitigated) to[loop above] node[above] {$o_m^{\rfct} \setminus \overline{rnd}$} (mitigated);
  \draw [red,->,arrows=-latex,thick] (a2) ..controls (247.62bp,110.37bp) and (246.81bp,129.0bp)  .. node[midway,sloped,fill=white!100] {$\tau$} (mitigated);
  \draw [red,->,arrows=-latex,thick] (a2) ..controls (242.68bp,99.262bp) and (210.33bp,80.672bp)  .. node[midway,sloped,fill=white!100] {$\tau$} (inactive);
  \draw [red,->,arrows=-latex,thick] (a1) ..controls (213.8bp,6.9263bp) and (197.27bp,13.269bp)  .. node[midway,sloped,fill=white!100] {$\tau$} (inactive);
  \draw [red,->,arrows=-latex,thick] (a2) ..controls (263.86bp,105.54bp) and (303.06bp,86.486bp)  .. node[midway,sloped,fill=white!100] {$\tau$} (active);
  \draw [->,arrows=-latex,thick] (inactive) to[loop left] node[left] {$o_n^{\rfct} \setminus rnd_0$} (inactive);
\end{tikzpicture}

   \caption{Nondeterministic risk factor $\rfct$
    \label{fig:riskfactor:ndet}}
\end{figure}

\subsection{Risk Spaces}
\label{sec:risk:space}

Risk factors give rise to further concepts: \emph{risk states, risk
  spaces, mitigation orders}, and \emph{risk structures}.
Let $n\in\mathbb{N}$ and $F = \{\rfct_i\mid i\in[1..n] \} \subset \mathbb{F}$ be a
finite set of risk factors~(\Cref{def:riskfactor}) where
$\rfct_i = (\mathit{Ph}_i, \Sigma_i^{\rfct_i}, \rightarrow_i,
\preceq_{\rfct_i}, c_i)$.

\begin{definition}[Risk State]
  \label{def:riskstate}
  Assume that risk factors are unique, that is,\xspace
  $i \neq j \Rightarrow \rfct_i \neq \rfct_j \land \mathit{Ph}_i \cap
  \mathit{Ph}_j = \varnothing$.
  Then, a \emph{risk state} is a faithful total injection
  $\sigma \colon F \rightarrow \bigcup_{\rfct \in F} \mathit{Ph}_{\rfct}$, that is,\xspace
  $\forall \rfct \in F \colon \sigma(\rfct)\in\mathit{Ph}_{\rfct}$.
\end{definition}

Observe that from \Cref{fig:riskfactor} follows
$\forall \rfct, \rfct[g] \in F \colon \Sigma^{\rfct} =
\Sigma^{\rfct[g]}$, that is,\xspace all risk factors correspond to processes with
the same alphabet, call it $\Sigma^F$.  This has no influence on risk
space composition as described below.  However, from this follows that
the corresponding CSP processes are composed in parallel synchronously
by $\parallel$ such that the process underlying each risk factor is
always ready to agree with some event of the environment.\footnote{In
  testing, this is sometimes called
  input-enabled %
  \citep{Tretmans2008-ModelBasedTesting}.} %
This construction guarantees deadlock freedom.

A risk state abstracts from states of a process
$P$~(\Cref{sec:kripke-struct-risk}) by focusing on risk-related
information in form of state propositions associated with the phases
of the risk factors.

\begin{definition}[Risk Space]
  \label{def:riskstatespace}
  For a set of risk factors $F$, a \emph{risk space} $R(F)$ is the
  function space
  given by
  \begin{equation*}
    \label{eq:riskstatespace}
    R(F) = \{ \sigma \in F \rightarrow \bigcup_{\rfct \in F} \mathit{Ph}_{\rfct}
    \mid 
    \sigma\; \text{is a total injection} \land
    \forall \rfct \in F \colon \sigma(\rfct)\in\mathit{Ph}_{\rfct} \}
  \end{equation*}
  We omit the parameter $F$ from $R$ if it is clear from the
  context and denote the set of all risk spaces by $\mathcal{R}$.
\end{definition}
Let
$phase\colon R \times \mathbb{F} \rightarrow \bigcup_{\rfct\colon\mathbb{F}}
\mathit{Ph}_{\rfct}$ be a map yielding the \emph{current phase} of a risk
factor $\rfct$ in a risk state $\sigma$.
The infix operator scheme
$\cdot\mid_{F'}\colon R(F) \times \mathbb{F} \rightarrow
R(F')$ describes a projection from the risk space
$R(F)$ to the risk space $R(F')$ where
$F'\subseteq F$.

\label{cnv:stateindex}%
We allow the convention
$\sigma^{(i)} = \sigma(\rfct_i) = phase(\sigma,\rfct_i)$ when
referring to the phase of the risk factor $\rfct_i$.
We can view $R$ as a set of ordered $n$-tuples~(formed by the
Cartesian product of $\mathit{Ph}_i$ after fixing some linear order over
$F$) or as a set of sets~(formed by all equivalence classes of phase
permutations over $F$).  These views permit the treatment of
$\sigma \in R$ as an unordered tuple or index set.
Particularly, for $i \neq j$, the tuples
$(\sigma^{(i)},\sigma^{(j)}) \in \mathit{Ph}_i\times\mathit{Ph}_j$ and
$(\sigma^{(j)},\sigma^{(i)}) \in \mathit{Ph}_j\times\mathit{Ph}_i$ identify
exactly the same risk state.  Consequently,
$(\sigma^{(i)},\sigma^{(i)}) =
(phase(\sigma,\rfct_i),phase(\sigma,\rfct_i))$ collapses to
$\sigma^{(i)}$.  In the following, one of the two views of $R$ will
occasionally be more convenient for the discussion.

\begin{remark}
  By \Cref{def:riskstatespace}, $R$ is non-empty and finite if and only if\xspace $F$
  is non-empty and finite.
  $R$ defines the set of all states an
  arbitrary\footnote{Below, we also handle $R$ as ``the most
    general~(risk) structure'' for a specific set of risk factors.}
  combination of phases of risk factors might give rise to.  Given a
  complete set of risk factors identified by the risk analyst for a
  specific application, only a small subset of $R$ might eventually be
  relevant for the machine in operation.  In general, we assume that
  identifying this subset is difficult.  Moreover, the
  \emph{relevance} of a risk state can be seen as a gradual quantity
  determined at run-time based on its context in $R$ and the process
  state. %
\end{remark}

\begin{definition}[Equality and Compatibility of Risk States]
  \label{def:risk:eqcomp}
  Two risk states $\sigma,\sigma'\in R$ are \emph{equal}, written
  $\sigma = \sigma'$, if and only if\xspace their corresponding phases are equal,
  formally, $\forall i\in[1..n]\colon \sigma^{(i)} = \sigma'^{(i)}$.
  Generally, given $F_1, F_2 \subset \mathbb{F}$, $\sigma\in R(F_1)$ and
  $\sigma'\in R(F_2)$ are \emph{compatible}, written
  $\sigma \approx \sigma'$, if and only if\xspace
  $\forall \rfct \in F_1 \cap F_2 \colon \sigma(\rfct) = \sigma'(\rfct)$,
  particularly, if $F_1 \cap F_2 = \varnothing$.  
\end{definition}
Consequently, state
equality %
implies state compatibility.

\begin{definition}[Risk Space Composition]
  Then, the \emph{composition}
  $\otimes\colon\mathcal{R}\times\mathcal{R} \rightarrow \mathcal{R}$ of
  the two risk spaces $R(F_1)$ and $R(F_2)$ is defined by
  \begin{equation}
    \label{eq:riskstatespace:comp}
    R(F_1) \otimes R(F_2) = \{
    \sigma_1 \cup \sigma_2
    \mid
    \sigma_1 \in R(F_1) \land \sigma_2 \in R(F_2)
    \land 
    \sigma_1 \approx \sigma_2
    \}.
  \end{equation}
\end{definition}
An analogous constraint is used for parallel composition of risk
structures below in \Cref{eq:stateconsistency}.  Now, we can derive a
basic law relating the union of risk factors and the composition of
risk spaces.  Furthermore, it will turn out that $R$ is a
homomorphism.

\begin{lemma}[Exchange of $\cup$ and $\otimes$]
  \label{lemma:riskstatespace:exchange}
  \[
    R(F_1 \cup F_2) = R(F_1) \otimes R(F_2)
  \]
\end{lemma}

\begin{proof}[Proof Sketch.]
  The proof is by mutual existence and uniqueness:
  For each $\sigma \in R(F_1 \cup F_2)$ (i) there exists a $\sigma_1
  \cup \sigma_2 \in R(F_1) \otimes R(F_2)$ and (ii) this pair is
  unique, and (iii, iv) vice versa.
  Details on the proof can be taken from
  \Cref{proof:riskstatespace:exchange}.
\end{proof}

\begin{lemma}[Homomorphism]
  \label{lem:homo1}
  $R$ is a homomorphism in the context of $(\mathbb{F},\cup)$ and $(\mathcal{R},\otimes)$.
  \begin{center}
    \begin{tikzcd}
      F_1,F_2 \ar{r}{\cup} \ar{d}[blue]{R}
      & F_1 \cup F_2 \ar{d}[blue]{R}\\
      R(F_1),R(F_2) \ar{r}{\otimes} & R(F_1 \cup F_2)
    \end{tikzcd}
  \end{center}
\end{lemma}
\begin{proof}[Proof Sketch.]
  We first make sure that we actually deal with semi-groups and then
  show by algebraic manipulation that $\otimes$ is associative.
  Details on the proof can be taken from \Cref{proof:homo1}.
\end{proof}

\paragraph{Special Risk States.}

We close this section with an analysis of specific classes of risk states.
$R(\varnothing)$ forms the \emph{empty risk space} and $R(\{\rfct\})$
the \emph{trivial risk space} for $\rfct$, leading to the following
law and equality:
\begin{corollary}
  \label{def:riskspace:misc}
  For any finite $F\subset\mathbb{F}$, $R(\varnothing)$ is the zero element of risk space
  composition with $\otimes$:
  \begin{align*}
    \label{eq:riskspace:zero}
    R(\varnothing) \otimes R(F)
    &= R(\varnothing) = \varnothing
    \tag{$\otimes$-zero}
    \\
    R(\{\rfct\})
    &= \{\rfct \mapsto \inact, \rfct \mapsto \activ, \rfct \mapsto \mitig\}
  \end{align*}
\end{corollary}

\begin{definition}[Locked Risk State]
  Based on \Cref{sec:risk:factors,def:riskstatespace} and given a set
  $F$ of $n$ risk factors, for any state $\sigma \in R$:
  \[
    \sigma \;\text{is \emph{risk-locked}} \iff %
    \forall i \in [1..n], \not\exists (p,e,p') \in\; \rightarrow_i
    \colon \sigma^{(i)} = p \land p \neq p'
  \]
  Otherwise, we call $\sigma$ \emph{risk-unlocked}.
\end{definition}

This notion is different from \emph{control stability} where the plant
has reached a stable state and from CSP's \emph{stable failures model}
$\mathcal{F}$ where stability refers to control states without invisible
internal events~(i.e.,\xspace $\tau$) and waiting for input.

States with an \emph{active final risk factor}
$\rfct$~(\Cref{sec:risk:factors}) take up a particularly bad role:
Such states, by definition, would expose the process $P$ to
\emph{residual risk infinitely long and often}, thus, making any
harmful consequences associated with $\rfct$ very likely.
On the one hand, such states, useful in modelling \emph{bad
  accidents}, are inevitable in any realistic risk model and, on the
other hand, a process $P$ should govern its~(probabilistic) choices
in order to not enter such states.

\section{Mitigation Orders}
\label{sec:risk:mitorders}

This section investigates various basic orders over risk spaces depending on
the qualitative and quantitative information available in the risk model.

\subsection{Qualitative Mitigation Orders}
\label{sec:qual-mitig-orders}

Let $R(F)$ be a risk space for a set of $n$ risk factors
$F \subseteq \mathbb{F}$ according to \Cref{def:riskstatespace}.
Again, we assume all risk factors are given indices in the range
$1..n$.  We use the convention of page~\pageref{cnv:stateindex} to
refer to parts of risk states by $\sigma^{(i)}$ with $i\in[1..n]$
and define a partial %
order $\preceq_{\mathsf{m}} \;\subseteq R \times R$ as
follows.

\begin{definition}[Fully Comparable Inclusive Mitigation Order]
  \label{def:fullincmitord}
  For any states $\sigma,\sigma' \in R$, define
  \[
    \sigma \preceq_{\mathsf{m}} \sigma' \iff \forall i \in [1..n]\colon
    \sigma^{(i)} \preceq_{\rfct_i} \sigma'^{(i)}
  \]
\end{definition}

By
$\sigma \prec_{\mathsf{m}} \sigma' \iff \sigma \preceq_{\mathsf{m}} \sigma' \land
\sigma \not=_{\mathsf{m}} \sigma'$, we induce the corresponding strict order.
$\sigma$ and $\sigma'$ are said to be \emph{incomparable} if and only if\xspace
$\sigma \not\preceq_{\mathsf{m}} \sigma' \land \sigma' \not\preceq_{\mathsf{m}}
\sigma$.  Intuitively, $\sigma \preceq_{\mathsf{m}} \sigma'$ signifies that
``$\sigma'$ is a better achievement in risk mitigation than
$\sigma$.''  However, note that $\preceq_{\mathsf{m}}$ requires full
comparability of two states.  It might be cumbersome to require such
comprehensive knowledge to determine which state is ``better or less
risky'' than another state.
Hence, in presence of irreducible~(i.e.,\xspace aleatory) uncertainty, we might
instead want to account for the partial knowledge in the orders of
risk factors' phases~(\Cref{def:riskfactor}) at the level of $R$ by
providing a relaxed partial order as follows.

\begin{definition}[Partially Comparable Inclusive Mitigation Order]
  \label{def:partincmitord}
  For states $\sigma,\sigma'\in R$, define
  \[
    \sigma \precsim_{\mathsf{m}} \sigma' \iff \forall i \in [1..n]\colon
    \sigma^{(i)} \preceq_{\rfct_i} \sigma'^{(i)} \lor
    \big((\sigma^{(i)},\sigma'^{(i)}) \not\in\; \preceq_{\rfct_i}
    \land
    (\sigma'^{(i)},\sigma^{(i)}) \not\in\; \preceq_{\rfct_i}\big)
  \]
\end{definition}
We use $\prec^{\sim}_{\mathsf{m}}$ and
$=_{\mathsf{m}}^{\sim}$ to distinguish the corresponding strict order
and equality for $\precsim_{\mathsf{m}}$ from $\prec_{\mathsf{m}}$.  Intuitively,
\Cref{def:partincmitord} requires a ``betterment in risk from $\sigma$
to $\sigma'$'' based exactly on the comparable phases.

\begin{lemma}
  \label{thm:risk:fullimplpartial}
  For any $\sigma,\sigma'\in R$,
  \[
    \sigma \preceq_{\mathsf{m}} \sigma' \Rightarrow \sigma \precsim_{\mathsf{m}} \sigma'
  \]
\end{lemma}
\begin{proof}[Proof of \Cref{thm:risk:fullimplpartial}.]
  $\preceq_{\mathsf{\rfct}}$ is antisymmetric.
  By definition of $\preceq_{\mathsf{m}}$, we may assume
  \begin{align*}
    & \forall i \in [1..n] \colon
      \sigma^{(i)} \preceq_{\rfct_i} \sigma'^{(i)}
      \tag{$\forall$-elim}
    \\
    & \vdash
      \sigma^{(i)} \preceq_{\rfct_i} \sigma'^{(i)}
      \tag{$\lor$-intro1}
    \\
    & \vdash
      \sigma^{(i)} \preceq_{\rfct_i} \sigma'^{(i)} \lor
      \big((\sigma^{(i)},\sigma'^{(i)}) \not\in\; \preceq_{\rfct_i} \land\;
      (\sigma'^{(i)},\sigma^{(i)}) \not\in\; \preceq_{\rfct_i}\big)
      \tag{$\forall$-intro, assumption for each $i$}
    \\
    & \vdash
      \forall i \in [1..n] \colon
      \sigma^{(i)} \preceq_{\rfct_i} \sigma'^{(i)} \lor
      \big((\sigma^{(i)},\sigma'^{(i)}) \not\in\; \preceq_{\rfct_i} \land\;
      (\sigma'^{(i)},\sigma^{(i)}) \not\in\; \preceq_{\rfct_i}\big)
  \end{align*}
\end{proof}

\begin{corollary}
  \label{cor:risk:uneqimpl}
  \[
    \sigma' \neq_{\mathsf{m}}^{\sim} \sigma
    \Rightarrow \sigma' \neq_{\mathsf{m}} \sigma
  \]
\end{corollary}
\begin{proof}[Proof of \Cref{cor:risk:uneqimpl}.]
    \begin{align}
    \notag
    \sigma' \neq_{\mathsf{m}}^{\sim} \sigma
    &\Rightarrow \sigma' \neq_{\mathsf{m}} \sigma
      \tag{by definition}
    \\
    \notag
    \neg(\sigma' \precsim_{\mathsf{m}} \sigma \land \sigma \precsim_{\mathsf{m}} \sigma')
    &\Rightarrow \neg(\sigma' \preceq_{\mathsf{m}} \sigma \land \sigma
      \preceq_{\mathsf{m}} \sigma')
      \tag{by conversion}
    \\
    \sigma' \precsim_{\mathsf{m}} \sigma \land \sigma \precsim_{\mathsf{m}} \sigma'
    &\Leftarrow \sigma' \preceq_{\mathsf{m}} \sigma \land \sigma
      \preceq_{\mathsf{m}} \sigma'
      \tag{by \Cref{thm:risk:fullimplpartial}}
  \end{align}
\end{proof}

\begin{corollary}[Converse of \Cref{thm:risk:fullimplpartial}]
  \label{cor:risk:fullimplpartial}
  \begin{align}
    \neg(\sigma \preceq_{\mathsf{m}} \sigma')
    &\Leftarrow \neg(\sigma \precsim_{\mathsf{m}} \sigma')
      \tag{by definition of $\not\preceq_{\mathsf{m}},\not\precsim_{\mathsf{m}}$}
    \\
    \sigma \not\preceq_{\mathsf{m}} \sigma'
    &\Leftarrow \sigma \not\precsim_{\mathsf{m}} \sigma'
      \tag{by negation and definition of $\preceq_{\mathsf{m}},\precsim_{\mathsf{m}}$}
    \\
    \label{eq:convfullimplpart}
    \sigma \succ_{\mathsf{m}} \sigma'
    \lor \big(
    (\sigma,\sigma') \not\in\; \preceq_{\mathsf{m}} \land\;
    (\sigma',\sigma) \not\in\; \preceq_{\mathsf{m}} \big)
    &\Leftarrow \sigma \succ^{\sim}_{\mathsf{m}} \sigma'
      \lor \big(
      (\sigma,\sigma') \not\in\; \precsim_{\mathsf{m}} \land\;
      (\sigma',\sigma) \not\in\; \precsim_{\mathsf{m}} \big)
  \end{align}
\end{corollary}
\begin{proof}[Proof of \Cref{cor:risk:fullimplpartial}.]
  Case 1: If
  $(\sigma,\sigma') \not\in\; \precsim_{\mathsf{m}} \land\; (\sigma',\sigma)
  \not\in\; \precsim_{\mathsf{m}}$ then~(by
  \Cref{def:fullincmitord,def:partincmitord}) also
  $(\sigma,\sigma') \not\in\; \preceq_{\mathsf{m}} \land\; (\sigma',\sigma)
  \not\in\; \preceq_{\mathsf{m}}$ and, therefore,
  \Cref{eq:convfullimplpart}.  \hfill\resizebox{.5em}{.4em}{$\boxslash$}

  Case 2: If $\sigma \succ^{\sim}_{\mathsf{m}} \sigma' = \sigma'
  \precsim_{\mathsf{m}} \sigma \land \sigma' \neq_{\mathsf{m}}^{\sim} \sigma$,
  then we have either
  $(\sigma,\sigma') \not\in\; \preceq_{\mathsf{m}} %
  \land\;
  (\sigma',\sigma) \not\in\; \preceq_{\mathsf{m}}$~(because of some incomparable phases)
  which fulfils
  \Cref{eq:convfullimplpart}. Alternatively, we have
  $(\sigma',\sigma) \in\; \preceq_{\mathsf{m}}$ which means $\sigma$ and
  $\sigma'$ are fully comparable and, because of
  \Cref{cor:risk:uneqimpl}, we have $\sigma \succ_{\mathsf{m}}
  \sigma'$. 
\end{proof}

\subsection{Quantitative Mitigation Orders}

So far, we have seen how partial orders account for a lack of
knowledge and potential uncertainties about risk.  Now, we will have a
look at how we can use \emph{severity information}, if available for
specific risk factors, to interpolate knowledge gaps, model
uncertainty, and derive a linear order over $R$.

We continue with further definitions to deal with intervals:
Given a family of intervals $I = ([l_i,u_i))_{i\in[1..n]}$ over
$\mathbb{R}_+$ with $l_i \leq u_i$, we define the \emph{convex hull}
of $I$ by a map $\cdot^*\colon
2^{\mathbb{R}^2_+}\rightarrow\mathbb{R}^2_+$ as follows:
\begin{equation*}
  I^* = [min \{l_i\}_{i\in[1..n]}, max \{u_i\}_{i\in[1..n]})
\end{equation*}
Furthermore, let
$active \colon R(F) \rightarrow 2^{F}$ with
$active(\sigma) = \{\rfct \in F \mid \sigma(\rfct) = \activ \}$ be the map
returning the set of active factors of a risk state.  Moreover, let
$S\colon R\to \mathbb{R}^2_+$ with
\begin{equation*}
  S(\sigma) = (s_{\rfct})_{\rfct \in active(\sigma)}^*
\end{equation*}
be a map for the construction of the severity interval of a risk
state.\footnote{Note that $S$ over-approximates~(i.e.,\xspace constructs the
  convex hull from) sparsely distributed severity intervals.}  Whereas
the \emph{minimum severity} of a risk factor $\rfct$ is given by the
interval $[0,0)$, the minimum severity of a risk state $S(\sigma)$ is
the empty set $\varnothing$.
For any
two real-valued intervals $[a,b),[c,d)\in\mathbb{R}^2_+$,
\citet{Ishibuchi1990-Multiobjectiveprogrammingoptimization} %
define with $[a,b) \leq [c,d) \iff a \leq c \land b \leq d$ a partial
order over such intervals. %

Moreover, we say that two risk states $\sigma,\sigma' \in R$ are
\emph{severity-equivalent} if and only if\xspace their accumulated severity
intervals are equal, that is,\xspace
$\sigma \sim_{\mathsf{s}} \sigma' \iff S(\sigma) = S(\sigma')$.  We have that
$\sigma = \sigma' \Rightarrow \sigma \sim_{\mathsf{s}} \sigma'$ because the
factors that are in their active phases are identical.
The relation $\sim_{\mathsf{s}}$ is an equivalence relation because it is
reflexive, symmetric, and transitive~(all by the usual equivalence
over intervals).  Furthermore, $\sim_{\mathsf{s}}$ induces equivalence classes
$[\sigma]_{\sim_{\mathsf{s}}} = \{\sigma'\in R \mid \sigma'\sim_{\mathsf{s}} \sigma \}$ over
$R$ for any $\sigma \in R$ with the corresponding quotient
class~$R/\sim_{\mathsf{s}}$.
With the \emph{severity intervals} $(s_i)_{i\in[1..n]}$, we now define
an order over $R/\sim_{\mathsf{s}}$.

\begin{definition}[Strong Mitigation Order]
  \label{def:strmitord}
  For $[\sigma]_{\sim_{\mathsf{s}}}, [\sigma']_{\sim_{\mathsf{s}}} \in R/\sim_{\mathsf{s}}$, define
  \begin{equation*}
    [\sigma]_{\sim_{\mathsf{s}}} \leq_{\mathsf{m}} [\sigma']_{\sim_{\mathsf{s}}} \iff
    \forall
    \sigma\in[\sigma]_{\sim_{\mathsf{s}}},\sigma''\in[\sigma']_{\sim_{\mathsf{s}}} \colon
    S(\sigma) \geq S(\sigma'')
    \lor
    S(\sigma'') \subset S(\sigma).
  \end{equation*}
  $[\sigma]_{\sim_{\mathsf{s}}} \leq_{\mathsf{m}} [\sigma']_{\sim_{\mathsf{s}}}$ can be dropped from
  $R/\sim_{\mathsf{s}}$, yielding
  \begin{equation}
    \label{eq:strmitord:dropped}
    \forall
    \sigma\in[\sigma]_{\sim_{\mathsf{s}}},\sigma''\in[\sigma']_{\sim_{\mathsf{s}}} \colon
    \sigma \leq_{\mathsf{m}} \sigma'' \iff
    S(\sigma) \geq S(\sigma'')
    \lor
    S(\sigma'') \subset S(\sigma).
  \end{equation}
\end{definition}
$\leq_{\mathsf{m}}$ codifies the circumstance that the risk state $\sigma'$ is
``better'' if the union of its severity intervals is
\begin{enumerate}[(a)]
\item \emph{lower} in the ranking $\leq$ of interval numbers or
\item \emph{narrower} than the corresponding union for $\sigma$.
\end{enumerate}
Condition (b) conveys the intuition that the interval carries less
uncertainty about the consequences expected from $\sigma'$ than from
$\sigma$.
Equivalence classes in $R/\sim_{\mathsf{s}}$ abstract from the risk factors from
which the merged severity intervals originate.  This abstraction has
to be carefully taken into account when using $\leq_{\mathsf{m}}$ and,
therefore, when specifying severity.  Note that $\leq_{\mathsf{m}}$ is based
on the convex hull of severity intervals from the the active phases of
a pair of risk states.  Apart from the convex hull, interval addition
and multiplication are relevant for alternative mitigation orders as
we shall see below.  However, a detailed investigation is left for
future work.
Let us now consider some core properties of $\leq_{\mathsf{m}}$.

\begin{lemma}
  \label{thm:risk:mitorderlinearity}
  $\leq_{\mathsf{m}}$ is linear %
  over $R/\sim_{\mathsf{s}}$.
\end{lemma}
\begin{proof}[Proof Sketch.]
  We show by case analysis that any two risk states are comparable and
  $\leq_{\mathsf{m}}$ is antisymmetric. The complete proof is stated in
  \Cref{sec:proofs}.
\end{proof}

\begin{corollary}
  After dropping \Cref{thm:risk:mitorderlinearity} by
  \Cref{eq:strmitord:dropped}, we have that $\leq_{\mathsf{m}}$ is also
  linear over $R$.
\end{corollary}

\begin{remark}[Method of Abstraction]
  \label{rem:abstroforder1}
  Severity intervals abstract from the potential consequences of risk
  factors.  The family $(s_i)_{i\in[1..n]}$ forms a \emph{cut of
    causal chains} and, hence, defines the \emph{scope} of the risk
  model.  This abstraction is left to the modeller~(i.e.,\xspace
  the risk analyst or safety engineer) and can vary significantly.
  Note, two different risk states
  $\sigma,\sigma' \in R$~(i.e.,\xspace $\sigma \neq \sigma'$) can well be
  severity-equivalent~(i.e.,\xspace $\sigma \sim_{\mathsf{s}} \sigma'$).
  Hence, the \emph{consequences} of several activated risk factors
  should be \emph{compatible} in the sense that the convex hull of the
  intervals of these factors maintains a consistent meaning of
  severity in~$(R,\leq_{\mathsf{m}})$.  
\end{remark}  
We will return to abstraction and compatibility of risk factors below
in \Cref{sec:risk:constraints} and instead continue here with the further
investigation of $\leq_{\mathsf{m}}$.

\begin{lemma}
  \label{thm:risk:mitwellorder}
  $(R,\leq_{\mathsf{m}})$ is well ordered.
\end{lemma}
\begin{proof}
  \Cref{thm:risk:mitwellorder} follows from a finite $R$~(by
  definition) and linearity of $\leq_{\mathsf{m}}$~(by
  \Cref{thm:risk:mitorderlinearity}). %
\end{proof}

\begin{definition}
  \label{def:orders:elements}
  For $\sigma \in R(F)$,
  we also write $\mathbf{0}^F \equiv \forall \rfct\in F\colon \sigma(\rfct) = \inact$
  and $\mathbf{F} \equiv \forall \rfct\in F\colon \sigma(\rfct) = \activ$.
  We denote by $\top^F$ the set of \emph{maximal elements} %
  and by $\bot^F$ the set of \emph{minimum elements} of %
  $(R,\preceq_{\mathsf{m}},\precsim_{\mathsf{m}},\leq_{\mathsf{m}})$.
  We characterise the minimal elements in $(R(F),\preceq_{\mathsf{m}})$ by
  \[
    \bot^F \equiv \{\sigma\in R(F)\mid \forall\sigma'\in R(F)\colon
    \sigma'\preceq_{\mathsf{m}}\sigma \Rightarrow \sigma'=_{\mathsf{m}}\sigma
    \}
  \]
  and analogously for $(R,\precsim_{\mathsf{m}})$ and $(R,\leq_{\mathsf{m}})$ and
  the maximal elements.
\end{definition}

\begin{corollary}
  \label{thm:orders:elements}
  If $\forall \rfct\in F\colon (Ph_{\rfct},\preceq_{\mathsf{\rfct}})$
  is linear, then $(R(F),\preceq_{\mathsf{m}}) = (R(F),\precsim_{\mathsf{m}})$.
  If $R \neq R(\varnothing)$ then $\bot^F$ and $\top^F$ are non-empty
  and, therefore, have a proper manifestation.
  For $(R/\sim_{\mathsf{s}},\leq_{\mathsf{m}})$, $\bot^F$ and $\top^F$ are singletons.
\end{corollary}
\begin{proof}[Proof of \Cref{thm:orders:elements}.]
  \label{proof:orders:elements}
  The proof is by contradiction.  For the sake of brevity, we only
  consider a sketch of this proof.  Assume we have two state classes
  $[\sigma]_{\sim_{\mathsf{s}}}, [\sigma']_{\sim_{\mathsf{s}}}$ in $\bot^F$ with
  $[\sigma]_{\sim_{\mathsf{s}}} \neq_m [\sigma']_{\sim_{\mathsf{s}}}$.  Because of our
  assumption, state classes are in linear order.  Thus, by definition
  of $\bot^F$, one of these state classes causes a violation of the
  universal quantification in \Cref{def:orders:elements} and,
  therefore, one of the classes cannot be in $\bot^F$ which
  contradicts our assumption.  Analogously for $\top^F$.
\end{proof}

\begin{corollary}
  \label{thm:strmitord:lattice}
  $(R/\sim_{\mathsf{s}},\leq_{\mathsf{m}})$ forms a complete lattice.
\end{corollary}
\begin{proof}[Proof of \Cref{thm:strmitord:lattice}.]
  \label{proof:strmitord:lattice}
  Linearity of $\leq_{\mathsf{m}}$ implies that every non-empty subset of
  $R/\sim_{\mathsf{s}}$ has a greatest lower bound
  and a least upper bound. %
\end{proof}

\subsection{Relating Mitigation Orders}
\label{sec:relat-mitig-orders}

Note that the strong mitigation order characterised by the
\Cref{thm:risk:mitorderlinearity,thm:risk:mitwellorder} is driven by
the number of active risk factors and their severity intervals
(because of the definition of $S$) but not by the equality of phases
of risk factors among the compared risk states.  This offers the
possibility of abstraction from individual risk factors and of
focusing on severity estimates.  To avoid infeasible models~(e.g.\xspace
specifications that get too strong to be realisable), we
\emph{require} that the addition of severity intervals constitutes a
\emph{relational extension}
of either $\preceq_{\mathsf{m}}$ or $\precsim_{\mathsf{m}}$, formally,
\[
  (\forall \sigma\in[\sigma]_{\sim_{\mathsf{s}}},\sigma''\in[\sigma']_{\sim_{\mathsf{s}}}
  \colon
  \sigma \preceq_{\mathsf{m}} \sigma'' \lor \sigma \precsim_{\mathsf{m}} \sigma'')
  \Rightarrow
  [\sigma]_{\sim_{\mathsf{s}}} \leq_{\mathsf{m}} [\sigma']_{\sim_{\mathsf{s}}}
\]
Dropped to $R$, this implies
$\sigma \preceq_{\mathsf{m}} \sigma'' \lor
\sigma \precsim_{\mathsf{m}} \sigma''
\Rightarrow
\sigma \leq_{\mathsf{m}} \sigma''$
for all pairs $(\sigma,\sigma'') \in
[\sigma]_{\sim_{\mathsf{s}}} \times [\sigma']_{\sim_{\mathsf{s}}}$, therefore,
\begin{equation}
  \label{eq:consistencyreq}
  \sigma \preceq_{\mathsf{m}} \sigma'' \lor \sigma \precsim_{\mathsf{m}} \sigma''
  \Rightarrow
  S(\sigma) \geq S(\sigma'') \lor S(\sigma'') \subset S(\sigma)
\end{equation}
for full and partial comparability, otherwise implying 
\begin{align}
  \label{eq:consistencyreq:nocomp} 
  (\sigma,\sigma'')\not\in\;\preceq_{\mathsf{m}}
  \land\;
  (\sigma,\sigma'')\not\in\;\precsim_{\mathsf{m}}
  & \Rightarrow \textsc{t}
\end{align}
Intuitively, if $\sigma$ is ``worse'' than $\sigma''$ then its
accumulated severity interval $S(\sigma)$ has to be greater than
that of $\sigma''$ and, therefore, must not be contained in that of
$\sigma''$.  Moreover, if $\sigma$ and $\sigma''$ are
incomparable in $\preceq_{\mathsf{m}}$ and $\precsim_{\mathsf{m}}$~(i.e.,\xspace some risk factors
have inversely ordered or incomparable phases) then $S(\sigma)$ and
$S(\sigma'')$ are allowed to form any relationship~(signified
by $\textsc{t}$ for ``true''), for example,\xspace \Cref{eq:consistencyreq}.

So, what is the~(necessary and) \emph{sufficient condition} on $F$ to
satisfy the requirement expressed by \Cref{eq:consistencyreq}?
Risk spaces and risk state pairs are the interpretations and,
therefore, potential models satisfying the relational extension
imposed by \Cref{eq:consistencyreq}.
The preparation of an answer to this question suggests the following
lemma:
\begin{lemma}
  \label{thm:risk:phaseinclusion}
  For $\sigma,\sigma' \in R(F)$,
  \begin{align*}
    \sigma \preceq_{\mathsf{m}} \sigma' \lor  %
    \sigma \precsim_{\mathsf{m}} \sigma'
    \Rightarrow
    active(\sigma) \supseteq active(\sigma').
  \end{align*}
\end{lemma}
\begin{proof}[Proof Sketch.]
  The proof is by induction over $F$ and relies on the assumption
  that, for any $\rfct\in F$, $\activ$ is the \emph{unique maximal element} in
  $\preceq_{\mathsf{\rfct}}$ and that $active$ only returns such elements.  The whole
  proof is stated in \Cref{proof:risk:phaseinclusion}.
\end{proof}

Again, fix a finite $F$ and a pair $\sigma,\sigma' \in R(F)$ and
assume
$\sigma \preceq_{\mathsf{m}} \sigma' \lor \sigma \precsim_{\mathsf{m}} \sigma'$.
Then, by \Cref{thm:risk:phaseinclusion}, $\sigma'$ incorporates
\emph{a subset of} active risk factors of $\sigma$.  $S(\sigma')$
maintains the right-hand part~(i.e.,\xspace $S(\sigma') \subset S(\sigma)$) of
the consequent of \Cref{eq:consistencyreq} in all of the following
cases:
\begin{enumerate}
\item no intervals are excluded~($\sigma'=\sigma$),
\item only intervals included in the others are excluded,
\item intervals only increasing the lower bound are excluded,
\item intervals only decreasing the upper bound are excluded, and
\item intervals increasing the lower bound and decreasing the upper
  bound are excluded.
\end{enumerate}
As a side note, \Cref{eq:consistencyreq} is also satisfied if all
factors in $F$ are assigned the same interval $c_0$.  In conclusion,
the sufficient condition on $F$ to satisfy \Cref{eq:consistencyreq} is
the ``unique maximal element'' precondition in the proof of
\Cref{thm:risk:phaseinclusion}.  Apart from this precondition,
\Cref{eq:consistencyreq} holds of an arbitrary finite
$F \subseteq \mathbb{F}$.
Below, we shall call $\preceq_{\mathsf{m}}$ and $\precsim_{\mathsf{m}}$ \emph{inclusive
  mitigation orders}, and $\leq_{\mathsf{m}}$ a \emph{strong mitigation order}.

\begin{theorem}
  The \emph{strong mitigation order} $\leq_{\mathsf{m}}$ extends the
  \emph{partially comparable inclusive mitigation order}
  $\precsim_{\mathsf{m}}$ which, in turn, extends the \emph{fully comparable
    inclusive mitigation order} $\preceq_{\mathsf{m}}$.  Formally, for
  $\sigma,\sigma' \in R$:
  \begin{equation*}
    \label{eq:strongExtInclMitOrder}
    \sigma \preceq_{\mathsf{m}} \sigma'
    \stackrel{\Cref{thm:risk:fullimplpartial}}{\Longrightarrow}
    \sigma \precsim_{\mathsf{m}} \sigma'
    \stackrel{\Cref{thm:risk:phaseinclusion}}{\Longrightarrow}
    \sigma \leq_{\mathsf{m}} \sigma'.
  \end{equation*}
\end{theorem}

\begin{remark}
  \label{rem:abstroforder2}
  Discrete and linear mitigation orders such as $\leq_{\mathsf{m}}$ promote
  machine implementations with \emph{negative utilitarian} decision
  ethics~\citep[p.~51]{Warburton2012-PhilosophyBasics}.  For example,\xspace
  severity intervals could be calculated at run-time based on sensor
  data about the possible operational situation of a system.  The
  expected outcomes of all enabled mitigation actions, if any, will
  then be comparable according to the resulting $\leq_{\mathsf{m}}$.  This
  comparability allows the assessment of the actual reachability of
  states with strictly lower risk.  Any resolution of a near-accident
  situation~(\Cref{exa:risk:roadtraffic}) or a \emph{tram
    problem}~\citep{Foot1978-ProblemAbortionDoctrine}~(also known
  as the ``trolley problem'') would then consist in the choice of the
  mitigation action leading to the state with the lowest risk or
  \emph{the least severe of the expected negative outcomes}.
  This scheme characterises negative utilitarianism.

  Linear orders globally resolve decisions based on explicit and,
  therefore, disputable criteria.  Consequently, utilitarian ethics
  have been criticised to lead to oversimplified approaches to resolve
  indecision.  According to
  \citet[p.~48f]{Warburton2012-PhilosophyBasics}, such critiques
  stress the difficulty of predicting the positive and negative
  effects of certain actions, in our case, the calculation of the
  severity of consequences of an activated risk factor.

  Fortunately, a structured risk model with dependencies between
  risk factors could be used to complement utilitarian decision ethics
  with \textsc{Kant}ian ethics, that is,\xspace to take conservative measures before
  high-severity risk factors get activated.
  For example,\xspace we can model the necessary preconditions of the tram problem as
  risk factors and a machine based on this model can use these factors
  to constrain its behaviour.

  Overall, although the presented model can be used with linear orders, the
  core discussions below try to stay agnostic of the mitigation order.
\end{remark}

\subsection{Local, Regional, and Global Safety}
\label{sec:risk:safety}

The orders $\preceq_{\mathsf{m}}, \precsim_{\mathsf{m}}$, and $\leq_{\mathsf{m}}$ are
\emph{local} in the sense that their definitions only require the
comparison of pairs of risk states.  Two more qualitative notions of
safety seem to be useful.

Let $R$ be non-empty and finite and
$reach\colon R \times \mathcal{P} \rightarrow 2^{R}$.  Given a
process $P \in \mathcal{P}$ and a risk state $\sigma \in R$,
$reach(\sigma,P)\subseteq R$ denotes the set of risk states reachable
from $\sigma$ by $P$ where $\sigma$ itself is always reachable and,
thus, $\sigma \in reach(\sigma,P)$.  Then, we use $\precsim_{\mathsf{m}}$ to
determine two non-empty sets of minimum and maximal elements in $R$
reachable in $P$, namely $max_{\precsim_{\mathsf{m}}}\{ reach(\sigma,P) \}$
and $min_{\precsim_{\mathsf{m}}}\{ reach(\sigma,P) \}$.

In a \emph{situation} represented by the process $P$ in a specific
risk state $\sigma$, these two sets signify the \emph{regionally
  safest}~($max$) and the \emph{regionally most hazardous}~($min$)
states, respectively.  The smallest such set will only and exactly
contain $\sigma$, representing ``the situation where $P$ cannot do
anything further about risk.''  This way, $\precsim_{\mathsf{m}}$ yields a
\emph{regional} notion of safety.  The notion is regional inasmuch as
once a maximal element in $max_{\precsim_{\mathsf{m}}}\{ reach(\sigma,P) \}$
is reached, the risk model allows no more reasoning about safer states
that $P$ could reach from $\sigma$ instead of maintaining its current
risk state.

In addition to local and regional safety,
$(R/\sim_{\mathsf{s}},\leq_{\mathsf{m}})$ yields a more \emph{global} notion
because of its linearity~(\Cref{thm:risk:mitorderlinearity}).  The
$\leq_{\mathsf{m}}$-based risk model is global in the sense that there are
always \emph{unique} safest and riskiest states in
$R/\sim_{\mathsf{s}}$~(\Cref{thm:orders:elements,thm:strmitord:lattice})
among all risk states reachable by $P$ from
$\sigma\in[\sigma]_{\sim_{\mathsf{s}}}$.  In contrast, the $\precsim_{\mathsf{m}}$-based
risk model will not guarantee uniqueness of the \emph{globally safest
  state} with respect to\xspace the reachable set of risk states.  Overall,
\Cref{thm:risk:mitwellorder} provides a necessary condition for
deriving \emph{strategies}~(i.e.,\xspace policies or choice resolutions)
that stabilise or terminate $P$ in these globally safest regions.
Note that the use of equivalence classes leads to more abstract forms
of safest and riskiest states.

These abstract $min/max$-bounded reachable sets can be used to
assess risk
at run-time, particularly, by estimating the probability of occurrence
and the severity of consequence from specific situational data only
available during operation and, then, by accumulating these
estimations according to the given risk model.
\emph{Safety of a process} then turns into the \emph{gradual presence
  and absence of the risk of undesired events} during operation,
during a system run, or, more formally, for subsets of $traces(P)$.
This notion complements safety as the crisp presence and absence of
undesired events, as the avoidance of risk factors~(e.g.\xspace undesired
events), and as the reduction of the probability of occurrence or the
severity of consequences of undesired events.
As a side note, the \emph{reliability of a process} as the
\emph{gradual presence and absence of defective behaviour of this
  process} can be seen as a special case of the presented approach
when considering only risk factors that model %
system faults.

\section{Dependencies between Risk Factors}
\label{sec:risk:constraints}

In risk analysis, we often wish to model \emph{causal relationships}
between~(phases of) risk factors and, consequently, risk states.  For example,\xspace
we might want to model that
\begin{enumerate}
\item the \emph{activation} of one risk factor \emph{causes} the
  \emph{activation} of another risk factor, 
\item the \emph{mitigation} of one risk factor \emph{causes} the
  \emph{activation} of another risk factor, or
\item the \emph{activation} of one risk factor \emph{requires} the
  \emph{activation} of another risk factor.
\end{enumerate}

\begin{example}
  \label{exa:depend1}
  Consider the following example illustrating these relationships: 
  \begin{enumerate}
  \item For example,\xspace water or oil on a robot's fingers~($\rfct_1$) \emph{cause}
    the robot's hand to be slippery~($\rfct_2$) such that holding a
    heavy object gets an
    action %
    with an increased likelihood of a negative
    outcome. %
  \item An increased grabbing pressure~($\rfct_3$) as a result to
    mitigate the object slipping out of the grabber's
    hold %
    could potentially \emph{cause} damage to the object.  This
    negative outcome can be modelled as a final risk factor $\rfct_4$.
  \item A damage to the object %
    ($\rfct_4$) \emph{requires %
      at least one of} high grabbing pressure~($\rfct_3$)
    or---applying backward identification of further risk factors---the
    object's high falling onto a hard surface~($\rfct_5$).  This
    fall~($\rfct_5$) %
    again \emph{requires at least one of} slippery fingers~($\rfct_1$)
    or a mistaken loosening of the grabber~($\rfct_6$).
  \end{enumerate}
\end{example}

\subsection{Relations over Risk Spaces}
\label{sec:constr-over-risk}

We can take account of relationships, such as illustrated in
\Cref{exa:depend1}, by imposing \emph{constraints} on pairs of risk
states and their comprising factors' phases.  For this, we use binary
relations over risk spaces $R$ to approximate \emph{causality
  assumptions} about parts of a process $P$ less known or less under
control~(typically, some kind of environment $E$) and \emph{causality
  requirements} to be imposed on parts of $P$ more known or more under
control.
In the following, we employ temporal logic~(\Cref{sec:temporal-logic})
and then relational specification to formalise constraints.

In the following, let $i,j\in[1..n]$ with $i\neq j$.  Now, consider
two distinct risk factors $\rfct_i,\rfct_j \in \mathbb{F}$.  For example,\xspace the
\emph{causes}~(or trigger) constraint\footnote{Used indirectly in form
  of \emph{minimum cut sequences} in FTA and directly in FMEA~(see
  \Cref{sec:relwork:randfail}).} can be defined as follows:
\begin{align}
  \label{eq:risk:constraint:causes:ltl}
  \rfct_i \;\mathtt{causes}\; \rfct_j
  & \equiv
  \square[\activ[\rfct_i]\Rightarrow\Diamond_{\leq
    t}(\activ[\rfct_j] \mathrel{\mathsf{U}} \neg\activ[\rfct_i])]
\end{align}
Note that in our model we have
$\neg\activ[\rfct_i] \Leftrightarrow \mitig[\rfct_i] \lor \inact[\rfct_i]$.
\Cref{eq:risk:constraint:causes:ltl} requires that for any path
through $\mathfrak{R}$ from a state in $R_0\subseteq R$ and at any step of this
path, if $\rfct_i$ is active then, within at most $t$ time units,
$\rfct_j$ must be active until $\rfct_i$ gets either inactive or
mitigated.  $\rfct_j$ can stay active forever and may already be
active before the activation of $\rfct_i$.
\begin{align}
  \label{eq:risk:constraint:causes:ltl:simple}
  \rfct_i \;\mathtt{causes}\; \rfct_j
  & \equiv
    \square[\activ[\rfct_i]
    \Rightarrow \circ(\activ[\rfct_j]
    \mathrel{\mathsf{U}} \neg\activ[\rfct_i])]
\end{align}
\Cref{eq:risk:constraint:causes:ltl:simple} forms a simplification of
\Cref{eq:risk:constraint:causes:ltl} taking into account the~(time)
abstraction used in our definition of $\mathfrak{R}$.  This abstraction implies
that the step of an inactive risk factor to its activated phase is a
logical time step, though, assuming a real-time duration of $>0$.

Let $\mathcal{C}$ be the set of all constraints.  For the translation
of TL constraints as described above into a form usable by parallel
composition, we use a map
$[\hspace{-.15em}[.]\hspace{-.15em}]_c \colon \mathcal{C} \rightarrow 2^{R \times R} $
to denote the relational semantics of constraints over $\mathfrak{R}$.
We will see in the following how applying a constraint
$c \in \mathcal{C}$ to a risk structure $\mathfrak{R}$ can restrict
the transition relation $\rightarrow$.

For example,\xspace \texttt{causes} constraints can be encoded as relations over $R$
as follows:
\begin{equation*}
  \label{eq:risk:constraint:causes}
  [\hspace{-.15em}[\rfct_i \;\mathtt{causes}\; \rfct_j]\hspace{-.15em}]_c
  =
  \{ (\sigma,\sigma') \in R \times R
  \mid
  \phase{i} = \activ[\rfct_i] \lor
  \phase[\sigma']{i} = \activ[\rfct_i]
  \Rightarrow
  \phase[\sigma']{j} = \activ[\rfct_j]\} 
\end{equation*}
\texttt{causes} extends to sets of risk factors.  Given
$F,F'\subseteq \mathbb{F}$ with $F \cap F' = \varnothing$, we define
\begin{align*}
  \label{eq:risk:constraint:causes:sets}
  [\hspace{-.15em}[F \;\mathtt{causes}\; F']\hspace{-.15em}]_c
  =
  \{ (\sigma,\sigma') \in R \times R 
  & \mid
    \exists \rfct_i \in F\colon
    \phase{i} = \activ[\rfct_i] \lor
    \phase[\sigma']{i} = \activ[\rfct_i]
    \Rightarrow
    \forall \rfct_j \in F' \colon \phase[\sigma']{j} = \activ[\rfct_j] \} 
\end{align*}
Note that all pairs $(\sigma,\sigma')$ violating the antecedent of the
conditional are in $[\hspace{-.15em}[\rfct_i \;\mathtt{causes}\; \rfct_j]\hspace{-.15em}]_c$ as
well.  Furthermore, this specific constraint allows \emph{immediate or
  weak causation} as well as \emph{delayed or strong causation} of at
most one transition~(that is,\xspace one ``logical'' step) in $\mathfrak{R}$.

The \texttt{requires} constraint can be written in TL form as follows:
\begin{equation}
  \label{eq:risk:constraint:requires:ltl}
  \rfct_i \;\mathtt{requires}\; \rfct_j
  \equiv
  \square[\neg\activ[\rfct_i] \mathrel{\mathsf{W}}_{\leq t} \activ[\rfct_j]] 
\end{equation}
and in relational form as follows:
\begin{equation*}
  \label{eq:risk:constraint:requires:relational}
  [\hspace{-.15em}[\rfct_i \;\mathtt{requires}\; \rfct_j]\hspace{-.15em}]_c
  =
  \{ (\sigma,\sigma') \in R \times R
  \mid
  \phase[\sigma']{i} = \activ[\rfct_i]
  \Rightarrow
  \phase{j} = \activ[\rfct_j]\} 
\end{equation*}
The lifting of \texttt{requires} to sets $F,F'\subseteq \mathbb{F}$ with
$F \cap F' = \varnothing$ is described as follows:
\begin{align*}
  \label{eq:risk:constraint:requires:sets}
  [\hspace{-.15em}[F \;\mathtt{requires}\; F']\hspace{-.15em}]_c
  =
  \{ (\sigma,\sigma') \in R \times R 
  & \mid
    \exists \rfct_i \in F\colon
    \phase[\sigma']{i} = \activ[\rfct_i]
    \Rightarrow
    \forall \rfct_j \in F' \colon \phase{j} = \activ[\rfct_j] \} 
\end{align*}
Note that this variant of the \texttt{requires} constraint refers to
\emph{all} factors specified on its left-hand side and, this way,
resembles an AND-gate as used in FTA.  The side condition
$F \cap F' = \varnothing$ rules out the case that a risk factor
requires or causes itself~(i.e.,\xspace the ``chicken and egg'' problem).

\begin{remark}
  The \texttt{causes} and \texttt{requires} constraints constitute
  basic templates for the design of constraints and exemplify how
  constraints can support the safety engineer in characterising
  causality in a state-based relational way.
\end{remark}

Analogously, there are many possible \emph{factor dependency
  constraints} over $\mathfrak{R}$ that resemble causal reasoning of techniques
such as FTA~(\Cref{sec:relwork:randfail}).  These constraints can be
classified along the following dimensions:
\label{def:constrdim}
\begin{itemize}
\item \emph{phase combination} (PC):
  active to active ($F \rightleftarrows F'$),
  active to mitigated ($F \rightleftarrows \overline{F'}$),
  active to inactive ($F \rightleftarrows 0^{F'}$),
  mitigated to mitigated ($\overline{F} \rightleftarrows \overline{F'}$),
  mitigated to active ($\overline{F} \rightleftarrows F'$),
  mitigated to inactive ($\overline{F} \rightleftarrows 0^{F'}$);
\item \emph{direction} (D) of cause-effect analysis:
  forward ($\rightarrow$, e.g.\xspace for modelling sufficient conditions or
  propagation), %
  backward ($\leftarrow$, e.g.\xspace for modelling necessary conditions or
  explanation);  %
\item \emph{polarity} (P): %
  obligation ($+$), %
  permission ($\circ$), %
  inhibition ($-$); %
\item \emph{causality} (C):
  strong (s), %
  weak (w); %
\item \emph{multiplicity} (M):
  one-to-one ($1:1$),
  one-to-many~($1:n$),
  many-to-one~($n:1$),
  many-to-many ($m:n$) where $m,n > 0$,
  self-referential\footnote{In this case, we allow $F = F' = \{\rfct\}$.} ($1:0$);
\item \emph{factor combination} (FC):
  ``${\sim}m$ out of $n$'' where $\sim \;\in \{\leq,=,\geq\}$ and $1 \leq m \leq n$.
\end{itemize}
These dimensions give rise to various types of
constraints.\footnote{Na\"ive combination of all dimensions
  ($6 * 2 * 3 * 2 * 5 * 3$) would result in 1080 constraints, many of
  them not essentially different, meaningful, or practical and, hence,
  resulting in significantly fewer.}
\Cref{tab:risks:constraints} describes types of constraints useful in
risk analysis.  Some of these constraints have been implemented in
\textsc{Yap}\xspace~\citep{Gleirscher-YapManual} which enables their use based on
previous discussions in
\citet{Gleirscher-YapManual,Gleirscher2018-Strukturenfurdie,Gleirscher2017-FVAV}.
However, their comprehensive treatment would exceed the scope of this
work.
As we have seen, constraints prune irrelevant state pairs and, this
way, determine the shape of $(R,\rightarrow)$.  This mechanism is
reflected by the following definition:

\begin{table}
  \small
  \caption{Overview of useful constraints. \textbf{Legend:} See 
    \emph{dimensions} on page~\pageref{def:constrdim}, Y \dots implemented in \textsc{Yap}\xspace.
    \label{tab:risks:constraints}}
  \begin{tabular}
    {%
    llcccccc
    p{8cm}}
    \toprule
    \textbf{Name}
    & \textbf{PC}
    & \textbf{D}
    & \textbf{P}
    & \textbf{C}
    & \textbf{M}
    & \textbf{FC}
    & \textbf{Y}
    & \small\textbf{Notes}
    \\\midrule
    causes
    & $F \rightleftarrows F'$ & $\rightarrow$ & + & w & 1:n & =n & \checkmark
    & The \emph{activation} of certain risk factors \emph{causes} the
      (successive) \emph{activation} of specific risk factors,
      \Cref{eq:risk:constraint:causes:ltl:simple}.
    \\
    causes$^{-1}$
    & $\overline{F} \rightleftarrows F'$ & $\rightarrow$ & + & w & m:n & =n & -
    & The \emph{mitigation} of certain risk factors \emph{causes} the
      \emph{activation} of other risk factors.
    \\\midrule
    requires
    & $F \rightleftarrows F'$ & $\leftarrow$ & + & w & 1:n & =n & \checkmark
    & The \emph{activation} of certain risk factors \emph{requires} the
      (prior) \emph{activation} of all out of a specified set of risk factors.
      AND-gate in FTA, \Cref{eq:risk:constraint:requires:ltl}.
    \\
    requires$_{1}$
    & $F \rightleftarrows F'$ & $\leftarrow$ & + & w & 1:n & $\geq$1 & -
    & The \emph{activation} of certain risk factors \emph{requires} the 
      (prior) \emph{activation} of at least \emph{one} out of a specified
      set of risk factors;
      OR-gate in FTA.
    \\
    prevents
    & $F \rightleftarrows F'$ & $\rightarrow$ & -- & w & 1:n & =n & \checkmark
    & The \emph{activation} of certain risk factors \emph{prevents} the 
      \emph{activation} of other risk factors
    \\
    preventsMit
    & $F \rightleftarrows \overline{F'}$ & $\rightarrow$ & -- & w & 1:n & =n & \checkmark
    & The \emph{activation} of certain risk factors \emph{prevents} the
      \emph{mitigation} of other risk factors
    \\\midrule
    excludes
    & $F \rightleftarrows 0^{F'}$ & $\rightarrow$ & + & w & 1:n & =n & \checkmark
    & the \emph{activation} of certain risk factors \emph{deactivates}
      (superposes or invalidates) other risk factors 
    \\\midrule
    direct
    & $F \rightleftarrows 0^{F'}$ & $\rightarrow$ & $\circ$ & s & 1:0 & =1 & \checkmark
    & the \emph{activation} of certain risk factors \emph{can} be
      immediately followed by \emph{their (!) deactivation}
    \\
    offRepair
    & $F \rightleftarrows 0^{F'}$ & $\rightarrow$ & -- & s & 1:0 & =1 & \checkmark
    & the \emph{activation} of certain risk factors \emph{cannot} be
      immediately followed by their \emph{deactivation}
    \\\bottomrule
  \end{tabular}
\end{table}

\begin{definition}[Relational Semantics for Constraints]
  \label{def:risk:constraints:sem}
  For a set of constraints $C \subseteq \mathcal{C}$,
  \begin{equation}
    \label{eq:risk:constraintsemantics}
    [\hspace{-.15em}[C]\hspace{-.15em}]_c =
    \left\{
      \begin{array}{ll}
        R \times R,
        & C = \varnothing
        \\
        \bigcap_{c\in C} [\hspace{-.15em}[c]\hspace{-.15em}]_c,
        & \text{otherwise}
      \end{array}
    \right.
    \quad \subseteq R \times R
  \end{equation}
\end{definition}
Constraints enable a top-down way of specifying risk models over risk
spaces.  In the context of the composition of risk spaces, this
alternative way requires the definition of the following healthiness
or well-formedness condition:

\begin{definition}[Well-formedness of Constraints]
  \label{def:risk:constraints:wf}
  Let $R(F)$ be a risk space formed by a family of risk factors
  $F = \{\rfct_i\}_{i\in[1..n]}$.  We say that a constraint
  $c \in \mathcal{C}$ is \emph{well-formed for $R(F)$} if and only if\xspace it does
  not refer to risk factors other than the ones in $F$, formally,
  if and only if\xspace $[\hspace{-.15em}[c]\hspace{-.15em}]_c\subseteq R(F) \times R(F)$.  We say that
  $C\subseteq\mathcal{C}$ is well-formed for $R(F)$ if and only if\xspace each
  element in $C$ is well-formed for $R(F)$.
\end{definition}

\subsection{Compatibility of Risk Factors}
\label{sec:comp-risk-fact}

Here, we continue with the methodological considerations from
\Cref{rem:abstroforder1} on consistent and meaningful interval and
probability specifications.

For example,\xspace assume to have two consequences $C_1$ and $C_2$ associated with
\emph{correct}\footnote{Correctness of $[l_1,u_1)$ and $[l_2,u_2)$
  assumes complete knowledge about consequences and their
  evaluation.} severity intervals $[l_1,u_1)$ and $[l_2,u_2)$.  Two
risk factors $\rfct_1$ and $\rfct_2$ can model three situations with their
\emph{individually estimated} intervals $\rfct_1.s$ and $\rfct_2.s$:
\begin{enumerate}
\item $\rfct_1$ and $\rfct_2$ share all their consequences, for example,\xspace
  $C_1$: If both factors get active, the convex hull can be backed
  by a meaningful consistency condition:
  $\{\rfct_1.s,\rfct_2.s\}^* \subseteq [l_1,u_1)$.
\item $\rfct_1$ and $\rfct_2$ do not share any consequences, for example,\xspace
  $\rfct_1$ models $C_1$ and $\rfct_2$ models $C_2$: If both factors
  get active, the convex hull extends both the range of consequences
  and severities and the consistency condition
  $\{\rfct_1.s,\rfct_2.s\}^* \subseteq \{[l_1,u_1),[l_2,u_2) \}^*$
  seems inappropriate.
  For example,\xspace if $C_1$ and $C_2$ signify the damage of two independent
  objects with the same interval, the convex hull would not account
  for this because of idempotency of interval union.
\item $\rfct_1$ and $\rfct_2$ share some of their consequences: For example,\xspace
  $\rfct_1$ potentially damages objects $A$ and $B$ and $\rfct_2$
  potentially damages objects $B$ and $C$. If both factors get
  active, the convex hull merges information about all consequences
  and, hence, results in a combination of the cases 1 and 2.
\end{enumerate}
While we can abstract from consequences shared by all risk factors,
the treatment of partially shared consequences requires more care.
The cases 2 and 3 can be modelled by an additional risk factor
$\rfct_3$ that is caused by $\rfct_1$ \emph{and} $\rfct_2$ and
carries an up-shifted severity interval, for example,\xspace
$\rfct_3.s = \rfct_1.s + \rfct_2.s$.  Below in
\Cref{sec:risk:constraints}, we will discuss how such a dependency
can be specified by constraints on the risk space, that is,\xspace by
$\{\rfct_1,\rfct_2\}$ \texttt{causes} $\{\rfct_3\}$ and $\{\rfct_3\}$
\texttt{excludes} $\{\rfct_1,\rfct_2\}$~(described in
\Cref{tab:risks:constraints}).

Based on this analysis, we call a set $F$ of risk factors
\emph{compatible} if there is a meaningful way of combining the
severity intervals for each subset $F' \subseteq F$.
For methodological support in achieving and maintaining compatibility,
the next paragraph exemplifies formal considerations of the consistency of
severity intervals.

\paragraph{Factor Characteristics and Dependencies.}

Constraints over risk spaces typically have implications on the
characteristics of risk factors such as severity intervals:
\begin{itemize}
\item For example,\xspace for a constraint $\rfct_1 \;\mathtt{causes}\; \rfct_2$, it is
  reasonable to claim $\rfct_1.s \supseteq \rfct_2.s$.
  
\item Analogously, for a constraint
  $\rfct_1 \;\mathtt{requires}\; \rfct_2$, we might wish to see
  $\rfct_1.s \subseteq \rfct_2.s$.
\end{itemize}
An in-depth analysis of techniques for making $F$ compatible and a
detailed discussion of a complete set of rules relating factor
characteristics and factor dependencies are out of scope of this work.

\section{Risk Structures}
\label{sec:risk-structures}

\emph{Risk spaces} are an abstract domain that can be used to equip
processes with \emph{risk awareness} and for the assessment of
\emph{risk mitigation capabilities} of such processes.  $R$
represents the \emph{unconstrained} composition of risk factors.
Events and transitions between risk states are ignored.  Instead of
considering $R$ in its full extension, we can \emph{select
  subsets} of $R$ for risk models of specific applications.
Even if not entirely known in advance, risk in a specific application
will usually have a \emph{specific structure} that might be more
adequately represented by a constrained \emph{composition} of risk
factors, eventually taking into account events and transitions between
risk states.  Specifically, based on \emph{constraints} on the
combinations of the risk factors' phases and on the
\emph{synchronisation} of events corresponding to the CSP model of
concurrency, we specify 
\begin{itemize}
\item which region of $R$ we want to pay attention to~(i.e.,\xspace scope of
  safety guarantees) and
\item which region of $R$ we consider to be safe for a process
  $P$~(i.e.,\xspace conventional safety).
\end{itemize}
In the following, we make use of risk factors' transition relations,
define a form of parallel composition, and discuss \emph{consistency,
  well-formedness, and validity} of risk structures.

\begin{definition}[Risk Structure]
  \label{def:riskstructure}
  A \emph{risk structure} $\mathfrak{R}$ is an expression of the
  form
  \[
    \mathfrak{R},\mathfrak{R}_1,\mathfrak{R}_2 ::= p \mid \mathfrak{R}_1 \parallel \mathfrak{R}_2 \mid [\mathfrak{R}]_C
  \]
  where $p \in \mathit{Ph}_{\rfct}$ (i.e.,\xspace
  $\inact, \activ, \mitig$) for any risk factor
  $\rfct \in F \subseteq \mathbb{F}$~(\Cref{def:riskfactor}) is an
  \emph{atom} and $C \subseteq \mathcal{C}$ is a set of
  \emph{constraints}~(\Cref{def:risk:constraints:sem}).
  $\mathcal{S}$ denotes the set of all risk structures,
  consequently, including the set of all phases of risk factors
  $\bigcup_{\rfct\in\mathbb{F}}\mathit{Ph}_{\rfct} \subset \mathcal{S}$.
\end{definition}
The operator $\parallel$ signifies the parallel composition of two
risk structures, and the operator $[\cdot]_C$ applies all constraints
in $C$ to a risk structure.  The semantics and algebraic properties of
these operators will be discussed below.

\paragraph{Operational Semantics of Risk Structures.}

With $\mathfrak{R}\in\mathcal{S}$ we associate a LTS
$[\hspace{-.15em}[\mathfrak{R}]\hspace{-.15em}]_{r} = (R, \Sigma, \rightarrow, R_0)$ with
\begin{itemize}
\item the risk space $R$~(\Cref{def:riskstatespace}),
\item the event set $\Sigma$~(\Cref{def:process}),
\item the transition relation
  $\rightarrow\; \subseteq R \times (2^{\Sigma}\setminus \varnothing)
  \times R$~(\Cref{def:risk:factorcomp}), and
\item a set of initial states $R_0 \subseteq R$.
\end{itemize}
Furthermore, let
$scope\colon \mathcal{S} \rightarrow 2^{\mathbb{F}}$ be a map
that identifies the risk factors referred to by a risk structure.  The
operational semantics of each construct of the language in
\Cref{def:riskstructure} are provided below.

\subsection{Atoms}

The operational semantics of a single risk factor $\rfct$ in phase
$p\in\mathit{Ph}_{\rfct}$ is given by
$[\hspace{-.15em}[p]\hspace{-.15em}]_{r} =
(\mathit{Ph}_{\rfct},\Sigma^{\rfct},\rightarrow_{\rfct},\{p\})$ and the
elements of this tuple are given by \Cref{def:riskfactor} and
described in \Cref{fig:riskfactor}.

\subsection{Parallel Composition}

Let $A,B,F \subseteq \mathbb{F}$ be sets of risk factors with
$A,B\subseteq F$ and, for $i\in\{1,2,3\}$, $\mathfrak{R}_i\in\mathcal{S}$
be arbitrary risk structures according to \Cref{def:riskstructure}
with $[\hspace{-.15em}[\mathfrak{R}_i]\hspace{-.15em}]_{r} = (R_i,\Sigma^u_i,\rightarrow_i,R_{i,0})$.

For situations where several safety engineers are trusted with the
task of risk analysis of a complex safety-critical system, we define
what it means to combine two risk structures with intersecting sets of
risk factors, formally,
$scope(\mathfrak{R}_1) \cap scope(\mathfrak{R}_2) \neq \varnothing$.  Because a single
risk factor cannot be in two different phases at the same time, we
employ a corresponding constraint to uniquely define the transition
relation resulting from parallel composition: \emph{only those risk
  states can be combined whose risk factors in the shared scope are in
  their identical phases.}
According to the discussion on
page~\labelcpageref{def:riskstatespace}, the following predicate
encodes this constraint:
\begin{equation}
  \label{eq:stateconsistency}
  \forall \sigma_1 \in R_1, \sigma_2 \in R_2 \colon
  cons(\sigma_1,\sigma_2)
  \iff
  \forall f \in scope(\mathfrak{R}_1) \cap scope(\mathfrak{R}_2) \colon \sigma_1^{(f)}
  = \sigma_2^{(f)}
\end{equation}

\begin{definition}[Parallel Composition]
  \label{def:risk:factorcomp}
  $[\hspace{-.15em}[\mathfrak{R}_1 \parallel \mathfrak{R}_2]\hspace{-.15em}]_{r} = (R, \Sigma^u,
  \rightarrow, R_{0})$ where
  \begin{itemize}
  \item $R = R_1 \otimes R_2$~(see \Cref{eq:riskstatespace:comp,lemma:riskstatespace:exchange}),
  \item
    $\rightarrow \; \subseteq R \times (2^{\Sigma} \setminus
    \varnothing) \times R$,
  \item the \emph{used alphabet}
    $\Sigma^u = \{ e \in 2^{\Sigma} \mid \exists \sigma,\sigma'
    \in R \colon \tran{\sigma}{e}{\sigma'} \}$, and
  \item $R_{0} = R_{1,0} \otimes R_{2,0}$.
  \end{itemize}
  For risk states
  $\sigma_1,\sigma_1' \in R_1, \sigma_2,\sigma_2' \in R_2$, the
  composed transition relation $\rightarrow$ is given by the
  following step rules:
  \begin{equation}
    \label{rule:risk:parcomp-l}
    \begin{prooftree}
      \hypo{\tran[1]{\sigma_1}{e}{\sigma_1'}}
      \hypo{\ntran[2]{\sigma_2}{}{}}
      \infer2[$[cons(\sigma_1,\sigma_2), cons(\sigma_1',\sigma_2)]$]{
        \tran{\sigma_1 \cup
          \sigma_2}{e}{\sigma_1' \cup \sigma_2}}
    \end{prooftree}
    \tag{$\parallel$\textsf{-l}, left}  
  \end{equation} 
  \begin{equation}
    \label{rule:risk:parcomp-r}
    \begin{prooftree}
      \hypo{\ntran[1]{\sigma_1}{}{}}
      \hypo{\tran[2]{\sigma_2}{f}{\sigma_2'}}
      \infer2[$[cons(\sigma_1,\sigma_2), cons(\sigma_1,\sigma_2')]$]{
        \tran{\sigma_1 \cup \sigma_2}{f}{\sigma_1 \cup
          \sigma_2'}}
    \end{prooftree}
    \tag{$\parallel$\textsf{-r}, right}
  \end{equation}
  \begin{equation}
    \label{rule:risk:parcomp-pl}
    \begin{prooftree}
      \hypo{\tran[1]{\sigma_1}{e}{\sigma_1'}}
      \hypo{\tran[2]{\sigma_2}{f}{\sigma_2'}}
      \infer2[$[cons(\sigma_1,\sigma_2), cons(\sigma_1',\sigma_2)]$]{
        \tran{\sigma_1 \cup \sigma_2}{e \setminus
          f}{\sigma_1' \cup \sigma_2}}
    \end{prooftree}
    \tag{$\parallel$\textsf{-pl}, partial left}
  \end{equation}
  \begin{equation}
    \label{rule:risk:parcomp-pr}
    \begin{prooftree}
      \hypo{\tran[1]{\sigma_1}{e}{\sigma_1'}}
      \hypo{\tran[2]{\sigma_2}{f}{\sigma_2'}}
      \infer2[$[cons(\sigma_1,\sigma_2), cons(\sigma_1,\sigma_2')]$]{
        \tran{\sigma_1 \cup \sigma_2}{f \setminus
          e}{\sigma_1 \cup \sigma_2'}}
    \end{prooftree}
    \tag{$\parallel$\textsf{-pr}, partial right}
  \end{equation}
  \begin{equation}
    \label{rule:risk:parcomp-b}
    \begin{prooftree}
      \hypo{\tran[1]{\sigma_1}{e}{\sigma_1'}}
      \hypo{\tran[2]{\sigma_2}{f}{\sigma_2'}}
      \infer2[$[cons(\sigma_1,\sigma_2), cons(\sigma_1',\sigma_2'), e \cap f \neq \varnothing]$]{
        \tran{\sigma_1 \cup \sigma_2}{e\cap f}{\sigma_1'
          \cup \sigma_2'}
      }
    \end{prooftree}
    \tag{$\parallel$\textsf{-b}, both}
  \end{equation}
\end{definition}
These step rules together resemble the step law of generalised
parallel composition in
CSP~\citep[Sec.~3.4]{Roscoe2010}
and can be used to determine the set of reachable states
$reach([\hspace{-.15em}[\mathfrak{R}_1\parallel\mathfrak{R}_2]\hspace{-.15em}]_{r}.R_0,P)$.

\begin{remark}
  \label{rem:parcomp}
  The side conditions in the rules of this composition operator
  prohibit any behaviour leading to inconsistent states.  In other
  words, the composition constrains the behaviour of two risk
  structures with an overlapping scope.  With \Cref{lem:constref}
  below, we will further discuss another form of behavioural
  constraint and the combination of composition and constraints for
  risk modelling.

  Were we to use arbitrary CSP processes as atoms and were we to allow
  different interfaces for each use of the parallel composition
  operator, several risk structures, when composed, would not
  guarantee freedom from interference, exhibit different event and
  state traces and, consequently, the order of their composition would
  lead to different risk models.
  Associative laws for generalised parallel
  composition in CSP are not universally applicable.
  The discussion by \citet[p.~60]{Roscoe2010} highlights how
  differences in the alphabets shared between each pair in a set of
  risk structures would entail meaning to the order in which these
  pairs are composed.  For example,\xspace in CSP, the equality
  $(P \genpar[X] Q) \genpar[Y] R = P \genpar[X] (Q \genpar[Y] R)$
  holds generally if $X = Y$.
  In case of $X \neq Y$, one has to compare the traces of the composed
  processes to prove specific guarantees to be preserved by their
  composition.  This can be computationally complex.
  More general forms of processes and composition have been dealt with
  in formalisms such as
  \textsf{\itshape Circus}\xspace~\citep{Oliveira2009,Oliveira2005-FormalDerivationState}
  and
  \textsc{Focus}~\citep{Broy2001-SpecificationDevelopmentInteractive}.
  
  Overall, we obtain two advantages from the use of risk factors as
  atoms in risk structures: This way, all atoms have exactly the same
  alphabet $\Sigma^F$, that is,\xspace all factors have to always agree on their
  view of the same overall process $P$.  Overlapping scopes in CSP
  terms then mean copies of risk factors that will be reduced by
  idempotency~(see \Cref{lemma:risk:rs:idem} below).
  The transformation of risk factors into CSP as described on
  page~\pageref{sec:risk:riskfactors-in-csp} encodes all information
  of the risk space $R$ into the processes $\inact, \activ$, and
  $\mitig$.  This way, we restrict the use of generalised
  parallel composition to the use of synchronous parallel
  composition. %
\end{remark}

\paragraph{Algebraic Properties of Composition ($\parallel$).}

The following discussion shows some properties desirable of the
(parallel) composition of risk structures.

\begin{lemma}[Idempotency of $\parallel$]
  \label{lemma:risk:rs:idem}
  For any $\mathfrak{R}_1 \in \mathcal{S}$, we have
  \begin{equation*}
    \mathfrak{R}_1 \parallel \mathfrak{R}_1 = \mathfrak{R}_1
    \tag{$\parallel$-idem}
  \end{equation*}
\end{lemma}
\begin{proof}[Proof of \Cref{lemma:risk:rs:idem}.]
  \label{proof:risk:rs:idem}
  From \Cref{def:riskstatespace}, by
  \Cref{lemma:riskstatespace:exchange},
  $R_1 = R_1 \otimes R_1$ is preserved.
  Using the convention from \Cref{sec:risk:space}, we can apply the
  rules from \Cref{def:risk:factorcomp} as follows:
  
  Because of the uniqueness of $\mathfrak{R}_1$, the composition takes two
  consistent views of $\mathfrak{R}_1$ offering identical states and events
  at each step.  So, because of $\sigma_1 = \sigma_2$ and $e = f$, the
  antecedents %
  in the rules $\parallel$\textsf{-l,r} can never be satisfied and the
  rules $\parallel$\textsf{-pl,pr} produce infeasible transitions
  to be pruned according to \Cref{def:riskfactor}. %
  The rule $\parallel$\textsf{-b} turns into a tautology
  and both views always exhibit an identical transition:
  \begin{equation*}
    \begin{prooftree}
      \hypo{\tran[1]{\sigma_1}{e}{\sigma_1'}}
      \infer1{\tran{\sigma_1 \cup \sigma_1}{e}{\sigma_1'
          \cup \sigma_1'}}
    \end{prooftree}
  \end{equation*}
  For $\sigma_1, \sigma_1' \in R_1$ and $e \in \Sigma^u_1$, this tautology
  preserves $\Sigma^u_1$, $\rightarrow_1$, and $R_{1,0}$.
\end{proof}

\begin{lemma}[Commutativity of $\parallel$]
  \label{lemma:risk:rs:comm}
  For any $\mathfrak{R}_1,\mathfrak{R}_2 \in \mathcal{S}$, we have
  \begin{equation*}
    \mathfrak{R}_1 \parallel \mathfrak{R}_2 = \mathfrak{R}_2 \parallel \mathfrak{R}_1
    \tag{$\parallel$-comm}
  \end{equation*}
\end{lemma}
\begin{proof}
  \label{proof:risk:rs:comm}
  \Cref{lemma:riskstatespace:exchange} in \Cref{sec:risk:space} makes
  it easy to show
  $R_1 \otimes R_2 = R(scope(\mathfrak{R}_1) \cup scope(\mathfrak{R}_2)) = R_2 \otimes
  R_1$.
  Furthermore, the symmetric duals~(i.e.,\xspace $\parallel$-r is the dual of
  $\parallel$-l, $\parallel$-pr is the dual of $\parallel$-pl) of the
  rules in \Cref{def:risk:factorcomp} yield the same $\rightarrow$
  and, hence, the same $\Sigma^u$.
\end{proof}

\begin{lemma}[Associativity of $\parallel$]
  \label{lemma:risk:rs:assoc}
  For any $\mathfrak{R}_1,\mathfrak{R}_2,\mathfrak{R}_3 \in \mathcal{S}$ with \emph{disjoint}
  scopes, that is,\xspace
  \[
    scope(\mathfrak{R}_1) \cap scope(\mathfrak{R}_2) \cap scope(\mathfrak{R}_3) = \varnothing,
  \]
  we have
  \begin{equation*}
    (\mathfrak{R}_1 \parallel \mathfrak{R}_2) \parallel \mathfrak{R}_3
    =
    \mathfrak{R}_1 \parallel (\mathfrak{R}_2 \parallel \mathfrak{R}_3)
    \tag{$\parallel$-assoc-1}
  \end{equation*}
\end{lemma}
\begin{proof}[Proof Sketch of \Cref{lemma:risk:rs:assoc}.]
  \label{proof:risk:rs:assoc}
  Because of empty shared scopes, the side conditions will always hold
  on both sides and states will be merged by disjoint union.  Set
  operations on states and events are commutative and associative.
  As pointed out in \Cref{rem:parcomp}, risk
  factors~(\Cref{def:riskfactor}) always share the whole alphabet
  $\Sigma^F$, that is,\xspace they synchronise on all
  events~(\Cref{sec:risk:space}).  Interleaving is avoided, the
  alphabetised parallel operator takes $\Sigma^F$ on both sides and,
  this way, reduces to synchronous parallel composition for which a
  general associative law is available.
\end{proof}
\Cref{lemma:risk:rs:assoc} allows the safe use of
$\parallel^{\rfct\in F} \inact $ as a shortcut for
$(( \dots (\inact[\rfct_1] \parallel \inact[\rfct_2]) \parallel \dots )
\parallel \inact[\rfct_n])$ with $\rfct_i\in F$.

\begin{lemma}[Associativity of $\parallel$]
  \label{lemma:risk:rs:assoc2}
  For any $\mathfrak{R}_1,\mathfrak{R}_2,\mathfrak{R}_3 \in \mathcal{S}$ with \emph{equal}
  scopes, that is,\xspace
  \[
    scope(\mathfrak{R}_1) = scope(\mathfrak{R}_2) = scope(\mathfrak{R}_3),
  \]
  we have
  \begin{equation*}
    (\mathfrak{R}_1 \parallel \mathfrak{R}_2) \parallel \mathfrak{R}_3
    =
    \mathfrak{R}_1 \parallel (\mathfrak{R}_2 \parallel \mathfrak{R}_3)
    \tag{$\parallel$-assoc-2}
  \end{equation*}
\end{lemma}
\begin{proof}[Proof Sketch of \Cref{lemma:risk:rs:assoc2}.]
  \label{proof:risk:rs:assoc2}
  The proof is analogous to the proof of \Cref{lemma:risk:rs:assoc}.
\end{proof}

\subsection{Constraints}
\label{sec:risk:constraints:appliedtors}

\Cref{sec:risk:constraints} introduced the concept of relations over
risk spaces with the aim of shaping~$(R,\rightarrow)$.  Based
on this concept, we now discuss how constraints can be embedded into
an operator of the language of risk structures as introduced in
\Cref{def:riskstructure}.  This way, constraints from a redundant
specification of \emph{beliefs about an application and operational
  risk associated with this application}.

Such redundancy can be used to identify inconsistencies between $\mathfrak{R}$
and the real world, potentially helpful in \emph{model refinement},
\emph{completion}, and \emph{validation}~\citep{Gleirscher2014a}.
Particularly, these inconsistencies allow choices for their
resolution.  We will make such inconsistencies explicit as follows.

For a set of risk factors $F \subseteq \mathbb{F}$ and a risk state
$\sigma \in R(F)$, we call the risk structure
\[
  \mathfrak{R}_{\sigma} =\, [\parallel^{\rfct\in F} phase(\sigma,\rfct)]_{C}
\]
with $R_0 = \{\sigma\}$ the \emph{characteristic risk
  structure} of $\sigma$.
From a risk state $\sigma$ or its characteristic structure
$\mathfrak{R}_{\sigma}$, we distinguish
three ways to determine the set of transitions $\tran{\sigma}{}{}$:\\
\begin{tabular}{llll}
  1.
  & $(\sigma,\sigma') \in [\hspace{-.15em}[C]\hspace{-.15em}]_c$
  & $\quad\Rightarrow \tran{\sigma}{\tau_c}{\sigma'} \in
    \tran[1]{\sigma}{}{}$,
  & derived from constraint $C$,
  \\
  2.
  & $\tran[]{\sigma}{e}{\sigma'}$
  & $\quad\Rightarrow
    \tran{\sigma}{e}{\sigma'} \in \tran[2]{\sigma}{}{}$,
  & derived from the \Cref{def:riskfactor,def:risk:factorcomp},
  \\
  3.
  & $e \in \mathit{initials}(P)$ and $\sigma' \in reach(\sigma,P)$
  & $\quad\Rightarrow
  \tran{\sigma}{e}{\sigma'} \in
  \tran[3]{\sigma}{}{}$,
  & derived from process $P$.
\end{tabular}

This framework gives rise to the following \emph{types of
  inconsistencies}:
\begin{enumerate}
\item The relation $\tran[1]{\sigma}{}{} \setminus\; \tran[2]{\sigma}{}{}$
  describes \emph{sensible} transitions with \emph{invisible events}
  signified by $\tau_c$.  We choose to prune transitions from
  $\rightarrow$ if they deviate from what is provided by the risk
  factors.
\item The relation $\tran[2]{\sigma}{}{} \setminus\; \tran[1]{\sigma}{}{}$ describes
  \emph{violent} transitions.  We choose to prune transitions from
  $\rightarrow$ if they violate constraints and, hence, lead to
  inconsistencies in $\mathfrak{R}$.
\item \label{case:risk:imperceptibletransitions}
  The relation $\tran[3]{\sigma}{}{} \setminus (\tran[1]{\sigma}{}{} \cap\; \tran[2]{\sigma}{}{})$
  describes \emph{imperceptible} transitions.  In $\rightarrow$, we
  choose to label such transitions with a $\tau_p$, making them
  subject of \emph{process-driven disclosure} %
  of $\mathfrak{R}$.
\item
  The relation $(\tran[1]{\sigma}{}{} \cap\; \tran[2]{\sigma}{}{}) \setminus \tran[3]{\sigma}{}{}$
  describes \emph{unrealised} transitions.  We choose to prune such
  transitions from $\rightarrow$, because they are not realised in $P$
  and, hence, would only add little value to $\mathfrak{R}$.
\end{enumerate}
This case analysis suggests several possibilities to design a semantics
for constraints.  As indicated above, we will investigate the
following transition relation $\rightarrow$:
\begin{align*}
  \tran{\sigma}{}{} 
  &= (\tran[1]{\sigma}{}{} \cup\; \tran[2]{\sigma}{}{}\; \cup\;
    \tran[3]{\sigma}{}{}) \setminus \big(
    \underbrace{(\tran[1]{\sigma}{}{} \setminus
    \tran[2]{\sigma}{}{})}_{\mbox{1. reduce to $F$}}
    \cup
    \underbrace{(\tran[2]{\sigma}{}{} \setminus
    \tran[1]{\sigma}{}{})}_{\mbox{2. reduce to $C$}}
    \cup
    \underbrace{(\tran[1]{\sigma}{}{} \cup\; \tran[2]{\sigma}{}{}
    \setminus \tran[3]{\sigma}{}{})}_{\mbox{4. reduce to $P$}}
    \big)
\end{align*}
Fortunately, we have
$\tran[1]{\sigma}{}{} \setminus\; \tran[2]{\sigma}{}{}\; = \varnothing$
because we require all risk factors to be input-enabled.

\begin{definition}[Constraint]
  \label{def:risk:constraint}
  Let $\mathfrak{R}$ be a risk structure~(\Cref{def:riskstructure}) with
  $[\hspace{-.15em}[\mathfrak{R}]\hspace{-.15em}]_{r} = (R,\Sigma^u,\rightarrow,R_0)$ and $C$ a set of
  constraints well-formed for
  $R$~(\Cref{sec:risk:constraints,def:risk:constraints:wf}).  Then,
  the \emph{constrained form} $[\mathfrak{R}]_C$ is defined by
  $[\hspace{-.15em}[[\mathfrak{R}]_C]\hspace{-.15em}]_{r} = (R, \Sigma^u_C, \rightarrow_C,
  R_0)$.
  $\rightarrow_C$ is determined by the
  following rule.  For any pair of risk
  states $\sigma,\sigma' \in R$
  \begin{equation}
    \label{rule:risk:constraint}
    \begin{prooftree}
      \hypo{\tran{\sigma}{e}{\sigma'}}
      \hypo{(\sigma,\sigma') \in [\hspace{-.15em}[C]\hspace{-.15em}]_c}
      \infer2{
        \tran[C]{\sigma}{e}{\sigma'}
      }
    \end{prooftree}
    \tag{$[\cdot]_C$-step}
  \end{equation}
\end{definition}
This definition incorporates the handling of inconsistencies according
to the cases 1 and 2.  The cases 3 and 4 can be helpful in the
incremental construction of $\mathfrak{R}$.

\paragraph{Algebraic Properties of Constraints ($[\cdot]_C$).}

Let $C_i \subseteq \mathcal{C}$ for $i \in \{1,2,3\}$.  We have that
$[\mathfrak{R}]_{\varnothing} = \mathfrak{R}$ by \Cref{def:risk:constraints:sem}.  The
following lemma incorporates that the order in which single
constraints are applied to $\mathfrak{R}$ does not matter.

\begin{lemma}[Exchange]
  \label{lem:risk:constraints:exchange}
  \begin{equation}
    \label{eq:risk:constraints:exchange}
    [[\mathfrak{R}]_{C_1}]_{C_2} = [\mathfrak{R}]_{C_1 \cup C_2}
  \end{equation}
\end{lemma}
\begin{proof}[Proof Sketch.]
  The proof is by induction over $C_1$, supported by an additional
  lemma for the induction step, and takes advantage of associativity of set
  intersection as used in \Cref{def:risk:constraints:sem}.  The
  detailed proof is stated in \Cref{proof:risk:constraints:exchange}.
\end{proof}

\begin{lemma}[Idempotency]
  \label{lem:risk:constraints:idem}
  \begin{align}
    [[\mathfrak{R}]_C]_C = [\mathfrak{R}]_C
    \tag{$[\cdot]_C$-idem}
  \end{align}
\end{lemma}
\begin{proof}
  This lemma follows from \Cref{lem:risk:constraints:exchange} and
  idempotency of set union.
\end{proof}

\begin{lemma}[Commutativity] %
  \label{lem:risk:constraints:comm}
  \begin{equation*}
    [[\mathfrak{R}]_{C_1}]_{C_2} = [[\mathfrak{R}]_{C_2}]_{C_1}
    \tag{$[\cdot]_C$-comm}
  \end{equation*}
\end{lemma}
\begin{proof}
  This lemma follows from \Cref{lem:risk:constraints:exchange} and
  commutativity of set union.
\end{proof}

\begin{lemma}[Associativity]
  \label{lem:risk:constraints:assoc}
  \begin{equation*}
    [[\mathfrak{R}]_{C_1 \cup C_2}]_{C_3} = [[\mathfrak{R}]_{C_1}]_{C_2 \cup C_3}
    \tag{$[\cdot]_C$-assoc}
  \end{equation*}
\end{lemma}
\begin{proof}
  This lemma follows from \Cref{lem:risk:constraints:exchange} and
  associativity of set union.
\end{proof}

For an arbitrary $C \subseteq \mathcal{C}$, we \emph{relax}
\Cref{def:risk:constraint} by establishing the following equivalence:
$[\mathfrak{R}]_C = [\mathfrak{R}]_{C'}$ for the largest $C'\subseteq C$ well-formed for
$R$~(\Cref{def:risk:constraints:wf}).  Then, $[\mathfrak{R}]_C$
denotes the risk structure resulting from applying only and exactly
the constraints in $C'$.  Moreover, \Cref{def:risk:constraint}
generalises \Cref{def:risk:factorcomp} by guarding the
$\parallel$-step.
\Cref{def:risk:constraint}, together with the relational pruning for
three out of four cases according the analysis on
page~\pageref{case:risk:imperceptibletransitions}, yields the
following general refinement law:
\begin{lemma}[Constraints Always Refine]
  \label{lem:constref}
  \begin{equation*}
    \mathfrak{R} \refinedBy\ [\mathfrak{R}]_C
  \end{equation*}
\end{lemma}
\begin{proof}
  The proof is by showing that
  $\traces{\mathfrak{R}} \supseteq \traces{[\mathfrak{R}]_C}$.  Fix
  $t \in \traces{[\mathfrak{R}]_C}$ with $t = f \mathop{\hat{}} l$ for induction
  over the split of $t$ into $f$ and $l$.
  Induction step: Assume that $f \in \traces{\mathfrak{R}}$ (IH).  With
  $l = \trace[e] \mathop{\hat{}} l'$ and $t = f \mathop{\hat{}} \trace[e] \mathop{\hat{}} l'$,
  there exist $\sigma,\sigma'\in R$ such that
  $\tran[C]{\sigma}{e}{\sigma'}$ and, according to
  \Cref{rule:risk:constraint}, such that $\tran{\sigma}{e}{\sigma'}$
  and $(\sigma,\sigma') \in [\hspace{-.15em}[C]\hspace{-.15em}]_c$.  Because of
  $\tran{\sigma}{e}{\sigma'}$, there must be a trace
  $f \mathop{\hat{}} \trace[e] \mathop{\hat{}} l'' \in \traces{\mathfrak{R}}$.  $e$ is part of
  $l$ and, therefore, $t$ such that we complete the induction step and
  establish a new IH.
  (We do not need to prove the equivalence $l' = l''$.)
  The case $f = \trace$ provides the induction start.
\end{proof}

\Cref{lem:constref} establishes a meaning of constraints that we might
usually expect from a methodological viewpoint.  The constraint
operator prunes the transition relation of $\mathfrak{R}$ according to a
relational specification of risk in terms of a set of constraints.

Finally, we discuss the special case that constraints are applied to
an atom $p\in\mathit{Ph}_{\rfct}$ associated with risk factor $\rfct$.  Observe that
$scope(p) = \{\rfct\}$ and, by \Cref{def:risk:constraints:wf}, $C$ is
well-formed for $R(scope(p))$ if $C$ contains only
constraints that refer to $\rfct$. Hence, our previous convention
entails $C' = \varnothing$ for the largest well-formed
$C'\subseteq C$.  From this observation, we obtain
\begin{align*}
  [p]_C
  \underbrace{= [p]_{C'}}_{\mbox{by convention}} \quad
  \underbrace{= [p]_{\varnothing}}_{\mbox{by \Cref{def:risk:constraints:wf}}} \quad
  \underbrace{= p}_{\mbox{by \Cref{def:risk:constraints:sem}}}.
\end{align*}

\section{Discussion}
\label{sec:rap:discussion}

Here, we briefly discuss potentials of risk structures to be used as a
formal artefact in failure assessments~(\Cref{sec:relwork:randfail})
and in the design of safety monitors~(\Cref{sec:relwork:monitors,sec:backg:monitors}).

\paragraph{Integration with Failure Analysis.}

As indicated in \Cref{sec:risk:constraints}, fault
trees~(\Cref{sec:relwork:randfail}) can be transformed into risk
structures by using dependencies.  Because of the step semantics
underlying constraints~(\Cref{def:risk:constraints:sem}), such a
translation is also possible for gates like \texttt{PAND} or
\texttt{POR}\footnote{\texttt{AND} and \texttt{OR} gates that express
  priority in a fault tree by taking into account the order of event
  occurrences.} as used in dynamic fault trees.  This way, the
presented approach offers the possibility to use fault trees generated
from architectures to be used as a risk structure or to be integrated
into an existing risk structure.  This leads to a practical way of
combining risk factors internal to an autonomous robot with risk
factors stemming from its operational environment.

\paragraph{Monitoring.}
\label{sec:notes-monitoring}

Each risk factor $\rfct$~(\Cref{fig:riskfactor}) can be implemented as
a \emph{monitor} of the process $P$.  Events of~$\rfct$ can be
translated into observers~(i.e.,\xspace sensors and variable checkers) of
transitions maintaining the monitor state and of transitions leading
to a phase change.  Consequently, implementing the events and phases
of each risk factor by observers allows the use of a risk structure
$\mathfrak{R}$ as a monitor of $P$'s risk state.

Given $\mathfrak{R}$ and $P$, we can devise an incremental and concurrent
approach to monitoring:  While the
situation %
of $P$ is monitored, $P$'s risk state is monitored and
continuously evaluated according to this situation.
Risk monitoring then comprises two parts: \emph{safety monitoring} for the
detection of endangerments~(i.e.,\xspace violations of safety properties), and
\emph{co-safety monitoring} for the detection of mitigations~(i.e.,\xspace
acceptances of mitigation properties).  For learning
$\tau_c$~(\Cref{sec:risk:constraints:appliedtors}), both monitoring
tasks can be carried through in an incremental mode, that is,\xspace introducing
non-existing states and transitions according the given states.

When using a risk factor as a monitor automaton for $P$, \emph{how
  do we deal with nondeterministic risk factors}?  In nondeterministic
risk factors, from an observed phase and after an observed event,
$P$ may have reached several distinct phases and, thus, several
distinct risk states.  Without further information the monitor does
not know which of these states has actually been reached by $P$.
The monitor would then be in a risk region which is ordered and
bounded by $min_{\leq_{\mathsf{m}}}/max_{\leq_{\mathsf{m}}}$ as discussed in
\Cref{sec:risk:safety}.  Given that the phases of risk factors carry
information about $P$'s state in terms of disjoint state invariants,
we can design state estimators into our monitor.  These estimators
might gradually be able to restore some of the lost state information
and again uniquely identify $P$'s actual risk state according to
$\mathfrak{R}$.

\section{Conclusions and Future Work}
\label{sec:concl-future-work}

The certification of robots and autonomous systems requires the
validation and verification of their controllers.  High automation
requires these controllers to have risk monitoring and handling
functions.  Complex machines and environments will make such functions
complex as well.  Such complexity might best be handled by incremental
construction and verification.  It is therefore helpful to have a
compositional method that allows these functions be incrementally
modelled and assessed at design time.  Eventually, these models will
be transformed into verified run-time monitors and mitigation
controllers that implement these functions.

In this work, we discussed a formal framework for analysing risk of an
autonomous system in its operational environment and for constructing
corresponding models that represent risk monitoring and handling
functions for this system.  Algebraic laws over these models support
systematic design and help the engineer to handle large models.  We
close the discussion with an explanation of how the risk model can be
used for verified monitor synthesis.

\paragraph{Future Work.}

Further steps in developing this framework will include
\begin{itemize}
\item the addition of probabilistic semantics to risk
  factors~(\Cref{sec:risk:factors}),
\item the extension of the presented framework to more general forms
  of risk factors~(\Cref{sec:risk:factors}),
\item the construction of risk lattices from risk spaces, mitigation
  orders, and event structures~(\Cref{sec:risk:mitorders}),
\item the development of reachability analysers for risk
  estimation %
  and synthesis of mitigation strategies, based on subsets of the risk
  space associated with a process state and
  dynamics~(\Cref{sec:risk:safety}), %
\item the use of a dynamical model of the process for the reachability
  analysers~(\Cref{sec:risk:safety}),
\item the enhancement of the discussion of factor dependency
  constraints~(\Cref{sec:risk:constraints}),
\item the investigation of distributivity of composition and
  constraints in risk
  structures~(\Cref{sec:risk:constraints:appliedtors}),
\item a mechanisation and extension of the given proofs in
  Isabelle/HOL.
\end{itemize}

\paragraph{Acknowledgements.}

This work is supported by the Deutsche Forschungsgemeinschaft~(DFG)
under Grants no.~GL~915/1-1 and GL~915/1-2.
I am deeply grateful to Jim Woodcock and Simon Foster for many
inspiring discussions, strongest guidance, particularly on the use of
algebraic methods, and for feedback on previous versions of this
manuscript.  It is my pleasure to thank Ana Cavalcanti and Cliff Jones
for raising important questions about the abstraction, composition,
and methodology underlying risk structures.  Furthermore, I owe
sincere gratitude to James Baxter, Alvaro Miyazawa, and Pedro Ribeiro
for many enlightening discussions and for contributing to an extremely
productive work environment at University of York.

\bibliographystyle{abbrvnat}  %
\bibliography{}

\appendix

\section{Nomenclature}
\label{sec:list-abbreviations}

See \Cref{tab:abbrev}.

\begin{table}
  \small
  \caption{Important abbreviations used in this article
    \label{tab:abbrev}}
  \begin{tabular}{ll}
    \toprule
    AV & Autonomous Vehicle\\
    BTA & Bow Tie Analysis\\
    CSP & Communicating Sequential Processes\\
    ETA & Event Tree Analysis \\
    FMEA & Failure Mode Effects Analysis \\ 
    FTA & Fault Tree Analysis \\
    HazOp & Hazard Operability (studies)\\
    LOPA & Layer Of Protection Analysis\\
    LTL & Linear Temporal Logic \\
    LTS & Labelled Transition System\\
    MCS & Minimum Cut Set\\
    PRA & Probabilistic Risk Assessment\\
    ROS & Robot Operating System \\
    STPA & System-Theoretic Process Analysis \\
    STAMP & System-Theoretic Accident Model and Process \\
    TL & Temporal Logic\\
    UML & Unified Modeling Language \\
    WBA & Why-Because Analysis \\
    \bottomrule
  \end{tabular}
\end{table}

\section{Proof Details}
\label{sec:proofs}

This section collects some of the more detailed proofs.

\begin{proof}[Proof of \Cref{lemma:riskstatespace:exchange}.]
  \label{proof:riskstatespace:exchange}
  The proof is by mutual existence and uniqueness:
  For each $\sigma \in R(F_1 \cup F_2)$ (i) there exists a $\sigma_1
  \cup \sigma_2 \in R(F_1) \otimes R(F_2)$ and (ii) this pair is
  unique, and (iii, iv) vice versa.

  We show (i): Let $\sigma \in R(F_1 \cup F_2)$, then, by definition,
  $\sigma$ is a total injection and, hence, every restriction of
  $\sigma$ is a total injection, particularly, the restrictions
  $\sigma|_{F_1}$ and $\sigma|_{F_2}$.  Obviously, we have
  $\sigma|_{F_1}(\rfct) = \sigma|_{F_2}(\rfct)$ for all $\rfct \in F_1 \cap F_2$.
  Furthermore, by definition, $\sigma(\rfct)$ is faithful to $\mathit{Ph}_{\rfct}$ for
  all $\rfct \in F_1 \cup F_2$ and, thus, so are both these restrictions.  These
  two results lead to $\sigma_1 = \sigma|_{F_1} \in R(F_1)$ and
  $\sigma_2 = \sigma|_{F_2} \in R(F_2)$ and, finally, to the existence
  of the wanted pair.
  \hfill\resizebox{.5em}{.4em}{$\boxslash$}

  We show (ii): Suppose there are two pairs
  $(\sigma_1,\sigma_2) \neq (\sigma_1',\sigma_2) \in R(F_1) \otimes
  R(F_2)$.  Then there exists $\rfct \in F_1 \setminus F_2$ where
  $\sigma_1(\rfct) \neq \sigma_1'(\rfct)$, thus, also
  $\sigma_1 \cup \sigma_2 \neq \sigma_1' \cup \sigma_2$.  However,
  there can be no faithful total injection $\sigma \in R(F_1 \cup F_2)$ such
  that
  $\sigma = \sigma_1 \cup \sigma_2 \land \sigma = \sigma_1' \cup
  \sigma_2$.
  \hfill\resizebox{.5em}{.4em}{$\boxslash$}

  We show (iii): For each
  $\sigma_1 \cup \sigma_2 \in R(F_1)\otimes R(F_2)$ there exists a
  $\sigma \in R(F_1 \cup F_2) = \{ \sigma \in (F_1 \cup F_2)
  \rightarrow \bigcup_{\rfct \in F_1 \cup F_2} \mathit{Ph}_{\rfct} \mid \sigma\;
  \text{is a total injection} \land \forall \rfct \in F_1 \cup F_2 \colon
  \sigma(\rfct)\in\mathit{Ph}_{\rfct} \}$: By definition, both $\sigma_1$ and
  $\sigma_2$ are faithful total injections~(i.e.,\xspace matching results for
  all $\rfct \in F_1 \cap F_2$).  Then, we can construct a faithful total 
  injection $\sigma$ by applying set union to the domains and
  co-domains of these two.  Thus, $\sigma \in R(F_1 \cup F_2)$.
  \hfill\resizebox{.5em}{.4em}{$\boxslash$}

  We show (iv): Suppose from $\sigma_1 \cup
  \sigma_2$ we can construct $\sigma \neq \sigma' \in R(F_1 \cup
  F_2)$ then there exists $\rfct \in F_1 \cup
  F_2$ such that $\sigma(\rfct) \neq \sigma'(\rfct)$.  However, then either
  $\sigma_1$ or $\sigma_2$ must have violated injectivity which
  in turn would have violated the definition of $R(F_1) \otimes R(F_2)$.
\end{proof}

\begin{proof}[Proof of \Cref{lem:homo1}.]
  \label{proof:homo1}
  $(\mathbb{F},\cup)$ is a semi-group because $\cup$ is an associative
  binary operation on $\mathbb{F}$.
  We have that
  \begin{equation}
    \sigma_1 \approx (\sigma_2 \cup \sigma_3) \Leftrightarrow
    \sigma_1\approx\sigma_2 \land \sigma_1 \approx
    \sigma_3.
    \label{frm:proof:lem:homo1}
  \end{equation}
  (The sub-proof based on the definition of $\approx$ is omitted here.)
  Based on \Cref{frm:proof:lem:homo1}, we show by algebraic
  manipulation that the binary operation $\otimes$ on $\mathcal{R}$ is
  associative:
  \begin{align*}
    & R(F_1) \otimes (R(F_2) \otimes R(F_3)) \\
    & = \{\sigma_1\cup\sigma\mid \sigma_1\in R(F_1)
      \land \sigma \in R(F_2) \otimes R(F_3) \land \sigma_1 \approx \sigma\} \\
    & = \{\sigma_1\cup(\sigma_2\cup\sigma_3)\mid \sigma_1\in R(F_1)
      \land (\sigma_2\cup\sigma_3) \in R(F_2) \otimes R(F_3) \land \sigma_1 \approx (\sigma_2\cup\sigma_3)\} \\
    & = \{\sigma_1\cup(\sigma_2\cup\sigma_3)\mid \sigma_1\in R(F_1)
      \land (\sigma_2\in R(F_2) \land \sigma_3 \in R(F_3) \land
      \sigma_2 \approx \sigma_3) \land \sigma_1 \approx (\sigma_2\cup\sigma_3)\} \\
    & = \{(\sigma_1\cup\sigma_2)\cup\sigma_3\mid (\sigma_1\in R(F_1)
      \land \sigma_2\in R(F_2)) \land \sigma_3 \in R(F_3) \land
      \sigma_2 \approx \sigma_3 \land \sigma_1 \approx
      (\sigma_2\cup\sigma_3)\}
      \tag{by \Cref{frm:proof:lem:homo1}} \\
    & = \{(\sigma_1\cup\sigma_2)\cup\sigma_3\mid (\sigma_1\in R(F_1)
      \land \sigma_2\in R(F_2)) \land \sigma_3 \in R(F_3) \land
      \sigma_2 \approx \sigma_3 \land \sigma_1\approx\sigma_2 \land
      \sigma_1 \approx \sigma_3\}\\
    & = \{(\sigma_1\cup\sigma_2)\cup\sigma_3\mid (\sigma_1\in R(F_1)
      \land \sigma_2\in R(F_2) \land \sigma_1 \approx \sigma_2) \land \sigma_3 \in R(F_3) \land
      (\sigma_1 \cup \sigma_2)\approx\sigma_3\} \\
    & = \{(\sigma_1\cup\sigma_2)\cup\sigma_3\mid
      (\sigma_1\cup\sigma_2)\in R(F_1) \otimes R(F_2) \land \sigma_3 \in R(F_3) \land
      (\sigma_1 \cup \sigma_2)\approx\sigma_3\} \\
    &
    = (R(F_1) \otimes R(F_2)) \otimes R(F_3)
  \end{align*}
  Hence, $(\mathcal{R},\otimes)$ is a semi-group, too.   
  \Cref{lemma:riskstatespace:exchange} then completes the proof.
\end{proof}

\begin{proof}[Proof of \Cref{thm:risk:mitorderlinearity}.]
  \label{proof:risk:mitorderlinearity}
  For this we only need to show that any two risk states
  $\sigma,\sigma'\in R$ are
  \begin{inparaenum}[(i)]
  \item comparable and
  \item antisymmetric:
    $[\sigma]_{\sim_{\mathsf{s}}} \leq_{\mathsf{m}} [\sigma']_{\sim_{\mathsf{s}}} \land
    [\sigma']_{\sim_{\mathsf{s}}} \leq_{\mathsf{m}} [\sigma]_{\sim_{\mathsf{s}}} \Rightarrow [\sigma]_{\sim_{\mathsf{s}}} =_{\mathsf{m}}
    [\sigma']_{\sim_{\mathsf{s}}}$.
  \end{inparaenum}

  We show (i) by showing that the conditions (a) and (b) guarantee the
  comparability of any two risk states in $R$ based on the interval
  order $\leq$ as defined above and use the fact that comparability
  can be dropped from %
  $R/\sim_{\mathsf{s}}$ to $R$ using
  $[\sigma]_{\sim_{\mathsf{s}}} \leq_{\mathsf{m}} [\sigma']_{\sim_{\mathsf{s}}} \iff
  \forall
  \sigma\in[\sigma]_{\sim_{\mathsf{s}}},\sigma''\in[\sigma']_{\sim_{\mathsf{s}}} \colon
  \sigma \leq_{\mathsf{m}} \sigma''$~(\Cref{eq:strmitord:dropped}).  We have
  to consider the following three 
  cases to complete the sub-proof of (i):

  \emph{Case ``no risk factors activated'':}
  $\varnothing \geq \varnothing \lor \varnothing \subset \varnothing$
  validates (b) and lack of comparability invalidates (a),
  yet we have $[\sigma]_{\sim_{\mathsf{s}}} \leq_{\mathsf{m}} [\sigma']_{\sim_{\mathsf{s}}}$.
  \hfill\resizebox{.5em}{.4em}{$\boxslash$}

  \emph{Case ``at least one risk factor activated only in $\sigma$'':}
  Let $(a,b) = S(\sigma)$.
  $(a,b) \geq \varnothing \lor (a,b) \subset \varnothing$ invalidates
  (a) and (b).  However, the observation
  $\varnothing \geq (a,b) \lor \varnothing \subset (a,b)$ yields
  $[\sigma']_{\sim_{\mathsf{s}}} \leq_{\mathsf{m}} [\sigma]_{\sim_{\mathsf{s}}}$.  The dual of this
  case works analogously.
  \hfill\resizebox{.5em}{.4em}{$\boxslash$}
  
  \emph{Case ``at least one risk factor activated in both $\sigma,\sigma'$'':}
  Let $(a,b) = S(\sigma)$ and $(c,d) = S(\sigma')$.  We need to show
  $(a,b) \geq (c,d) \lor (a,b) \subset (c,d)$ for all
  $a,b,c,d \in\mathbb{R}$: We can assume $a\leq b \land c\leq d$ by
  definition of intervals. Then, we face the following cases:
  \begin{itemize}
  \item $c \leq a \land d \leq b$ validates (a) by definition of $\leq$
    over intervals.
  \item $c > a \land d \leq b$ validates (b).
  \item $c \leq a \land d > b$ implies (b) for the dual case
    $[\sigma']_{\sim_{\mathsf{s}}} \leq_{\mathsf{m}} [\sigma]_{\sim_{\mathsf{s}}}$.
  \item $c > a \land d > b$ implies (a) for the dual case.
  \end{itemize}
  Hence, from each of these four cases, either $[\sigma]_{\sim_{\mathsf{s}}}
  \leq_{\mathsf{m}} [\sigma']_{\sim_{\mathsf{s}}}$
  or $[\sigma']_{\sim_{\mathsf{s}}} \leq_{\mathsf{m}} [\sigma]_{\sim_{\mathsf{s}}}$ follows.

  Then, we show (ii) by contradiction: Assuming
  $[\sigma]_{\sim_{\mathsf{s}}} \leq_{\mathsf{m}} [\sigma']_{\sim_{\mathsf{s}}} \land
  [\sigma']_{\sim_{\mathsf{s}}} \leq_{\mathsf{m}} [\sigma]_{\sim_{\mathsf{s}}}$, we have
  $\forall \sigma\in[\sigma]_{\sim_{\mathsf{s}}},\sigma''\in[\sigma']_{\sim_{\mathsf{s}}}
  \colon \sigma \leq_{\mathsf{m}} \sigma'' \land \sigma \geq_{\mathsf{m}} \sigma''$ by
  dropping from $R/\sim_{\mathsf{s}}$.  Now, we claim that
  $[\sigma]_{\sim_{\mathsf{s}}} \not=_{\mathsf{m}} [\sigma']_{\sim_{\mathsf{s}}}$.  Hence,
  $\exists \sigma\in[\sigma]_{\sim_{\mathsf{s}}},
  \sigma''\in[\sigma']_{\sim_{\mathsf{s}}}\colon \sigma \not=_{\mathsf{m}} \sigma''$ and,
  consequently,
  $\exists \sigma\in[\sigma]_{\sim_{\mathsf{s}}},
  \sigma''\in[\sigma']_{\sim_{\mathsf{s}}}\colon \sigma \not\leq_{\mathsf{m}} \sigma'' \lor \sigma
  \not\geq_{\mathsf{m}} \sigma''$.  The latter contradicts our assumption.
\end{proof}

\begin{proof}[Proof of \Cref{thm:risk:phaseinclusion}.]
  \label{proof:risk:phaseinclusion}
  Fix a finite $F \subseteq \mathbb{F}$ and a pair
  $\sigma, \sigma' \in R(F)$.  The proof is by induction over $F$ and
  relies on the assumption that, for any $\rfct\in F$, the phase $\activ$ is the
  \emph{unique maximal element} of $\preceq_{\mathsf{\rfct}}$ and that 
  $\mathit{active}$ only returns such elements:

  Induction start $F_0 = \varnothing$:
  $\sigma|_{\varnothing} \preceq_{\mathsf{m}} \sigma'|_{\varnothing}$ holds
  trivially and so does
  $active(\sigma'|_{\varnothing}) \subseteq
  active(\sigma|_{\varnothing})$.

  Induction step (IS) $F_{n+1} = F_n \cup \{\rfct\}$ where $n \geq 0$ and
  $\rfct \in F\setminus F_n$: For the induction hypothesis (IH), assume
  $\sigma|_{F_n} \preceq_{\mathsf{m}} \sigma'|_{F_n} \Rightarrow
  active(\sigma'|_{F_n}) \subseteq active(\sigma|_{F_n})$.  Then, we
  show that
  $\sigma|_{F_{n+1}} \preceq_{\mathsf{m}} \sigma'|_{F_{n+1}} \Rightarrow
  active(\sigma'|_{F_{n+1}}) \subseteq
  active(\sigma|_{F_{n+1}})$~(i.e.,\xspace the IS).
  
  For this, prove
  \begin{itemize}
  \item the case of \emph{incomparable phases
      $(\sigma(\rfct),\sigma'(\rfct)) \not\in\; \preceq_{\mathsf{\rfct}}$}: In this case, the
    state pair gets incomparable and the implication is trivially fulfilled,
  \item the \emph{inverse case} $\sigma(\rfct) \succ_{\rfct} \sigma'(\rfct)$: In
    this case, the state pair gets incomparable and the implication is again
    trivially fulfilled, and
  \item the \emph{aligned case} $\sigma(\rfct) \preceq_{\mathsf{\rfct}} \sigma'(\rfct)$: In
    this case, the state pair stays comparable.  However, $\sigma'(\rfct)$ can
    either be $\activ$~(hence, $\sigma(\rfct) = \activ$ and maintaining
    $\subseteq$), $\mitig$~(hence,
    $\sigma(\rfct) \in \{\activ,\mitig,\inact\}$ and maintaining
    $\subseteq$), or $\inact$~(see former sub-case).
  \end{itemize}
  Having proved these cases completes the induction step by
  establishing IS.

  For $\precsim_{\mathsf{m}}$, we only have to substitute case 1 of the IS: For
  $\rfct$, the smallest $\preceq_{\mathsf{\rfct}}$ after \Cref{def:riskfactor}
  implies incomparability at most for
  $(\inact,\mitig)$ and $(\mitig,\inact)$.  These two phase
  pairs of $\rfct$ would not alter $\mathit{active}$ of both states and hence
  maintain $\subseteq$ as well.
\end{proof}

\begin{proof}[Proof of \Cref{lem:risk:constraints:exchange}.]
  \label{proof:risk:constraints:exchange}
  From associativity of set intersection in
  \Cref{def:risk:constraints:sem}, we know that the order in which
  constraints are applied to $R \times R$ does not matter.
  Consequently, we also have
  $[\hspace{-.15em}[c]\hspace{-.15em}]_c \supseteq [\hspace{-.15em}[\{c,c'\}]\hspace{-.15em}]_c$ for any
  $c,c' \in \mathcal{C}$, that is,\xspace constraints only prune
  $R \times R$.
  First, we prove the following lemma:
  \begin{equation}
    \label{eq:risk:constraints:exchange:helper}
    [[\mathfrak{R}]_{\{c\}}]_{C} = [\mathfrak{R}]_{\{c\} \cup C}
  \end{equation}
  Fix $\tran[]{\sigma}{e}{\sigma'}$.

  $\Rightarrow$:
  \begin{itemize}
  \item If $\tran[]{\sigma}{e}{\sigma'} \in [\hspace{-.15em}[c]\hspace{-.15em}]_c$ then
    $[\cdot]_C$-step is applicable to $[\mathfrak{R}]_{\{c\}}$, and if
    $\tran[]{\sigma}{e}{\sigma'} \in [\hspace{-.15em}[C]\hspace{-.15em}]_c$ then to
    $[[\mathfrak{R}]_{\{c\}}]_C$.  Then, $\tran[]{\sigma}{e}{\sigma'}$ must be
    in $[\hspace{-.15em}[c]\hspace{-.15em}]_c \cap [\hspace{-.15em}[C]\hspace{-.15em}]_c$ which by
    \Cref{def:risk:constraints:sem} is $[\hspace{-.15em}[\{c\} \cup C]\hspace{-.15em}]_c$.
    Because of
    $\tran[]{\sigma}{e}{\sigma'} \in [\hspace{-.15em}[\{c\} \cup C]\hspace{-.15em}]_c$, 
    $[\cdot]_C$-step applies to $[\mathfrak{R}]_{\{c\} \cup C}$.
  \item If $\tran[]{\sigma}{e}{\sigma'} \not\in [\hspace{-.15em}[c]\hspace{-.15em}]_c$ or
    $\tran[]{\sigma}{e}{\sigma'} \not\in [\hspace{-.15em}[C]\hspace{-.15em}]_c$ then the
    $[\cdot]_C$-step is not applicable to $[[\mathfrak{R}]_{\{c\}}]_C$.
    Because of
    $[\hspace{-.15em}[\{c\} \cup C]\hspace{-.15em}]_c = [\hspace{-.15em}[c]\hspace{-.15em}]_c \cap [\hspace{-.15em}[C]\hspace{-.15em}]_c$, this also
    holds of $[\mathfrak{R}]_{\{c\} \cup C}$.
  \end{itemize}
  \hfill\resizebox{.5em}{.4em}{$\boxslash$}

  $\Leftarrow$:
  \begin{itemize}
  \item If $\tran[]{\sigma}{e}{\sigma'} \in [\hspace{-.15em}[\{c\} \cup C]\hspace{-.15em}]_c$
    then $[\cdot]_C$-step is applicable to $[\mathfrak{R}]_{\{c\} \cup C}$ and,
    because of $[\hspace{-.15em}[c]\hspace{-.15em}]_c \cap [\hspace{-.15em}[C]\hspace{-.15em}]_c$, twice to
    $[[\mathfrak{R}]_{\{c\}}]_C$.
  \item If
    $\tran[]{\sigma}{e}{\sigma'} \not\in [\hspace{-.15em}[\{c\} \cup C]\hspace{-.15em}]_c$ then
    $[\cdot]_C$-step is not applicable to $[\mathfrak{R}]_{\{c\} \cup C}$ and,
    because of
    $\tran[]{\sigma}{e}{\sigma'} \not\in [\hspace{-.15em}[c]\hspace{-.15em}]_c \cap
    [\hspace{-.15em}[C]\hspace{-.15em}]_c$, at most once to $[[\mathfrak{R}]_{\{c\}}]_C$.
  \end{itemize}
  \hfill\resizebox{.5em}{.4em}{$\boxslash$}
  
  Second, we show by induction over $C_1$ that the $[\cdot]_C$-step
  rule applies in the same way to both sides of
  \Cref{eq:risk:constraints:exchange}.
  Induction start with $C_1^0 = \varnothing$: In this case, we have
  $[[\mathfrak{R}]_{\varnothing}]_{C_2} = [\mathfrak{R}]_{\varnothing \cup C_2}$ and
  because $[\cdot]_C$-step is always applicable, 
  $[\mathfrak{R}]_{C_2} = [\mathfrak{R}]_{C_2}$.
  \hfill\resizebox{.5em}{.4em}{$\boxslash$}

  Induction step (IS) with $C_1^{n+1} = C_1^n \cup \{c\}$ where
  $n \geq 0$ and $c \in C_1 \setminus C_1^{n}$:
  By assuming $[[\mathfrak{R}]_{C_1^n}]_{C_2} = [\mathfrak{R}]_{C_1^n \cup C_2}$ (IH), we show
  \begin{align*}
    & [[\mathfrak{R}]_{C_1^{n+1}}]_{C_2}
      \tag{by def of IS}
    \\
    &= [[\mathfrak{R}]_{C_1^n \cup \{c\}}]_{C_2}
        \tag{by $\cup$-comm and \Cref{eq:risk:constraints:exchange:helper}}
    \\
    &= [[[\mathfrak{R}_{\{c\}}]_{C_1^n}]_{C_2}
        \tag{by IH}
    \\
    &= [[\mathfrak{R}_{\{c\}}]_{C_1^n \cup C_2}
        \tag{by \Cref{eq:risk:constraints:exchange:helper}}
    \\
    &= [\mathfrak{R}]_{\{c\} \cup C_1^n \cup C_2}
        \tag{by def}
    \\
    &= [\mathfrak{R}]_{C_1^{n+1} \cup C_2}
  \end{align*}
\end{proof}

\label{lastpage}
\end{document}